\documentclass[11pt]{article}
\usepackage[margin=1in]{geometry}
\usepackage{mathpazo}
\usepackage{color}
\usepackage{cite}
\usepackage{amsthm}
\usepackage{soul}
\usepackage{amsmath}
\usepackage{amssymb}
\usepackage{comment}
\usepackage{mdframed}
\usepackage{braket}
\usepackage{thmtools}
\usepackage[colorlinks = true]{hyperref}
\usepackage{cleveref}
\usepackage{graphicx}
\usepackage{mathtools}
\usepackage{booktabs}
\usepackage{dsfont}
\usepackage{float}
\usepackage{tablefootnote} 
\usepackage{orcidlink}
\usepackage[super]{nth}

\usepackage{xcolor}
\definecolor{darkred}{rgb}{0.5,0,0}
\definecolor{darkblue}{rgb}{0,0,0.5}
\definecolor{darkgreen}{rgb}{0,0.5,0}
\hypersetup{
  urlcolor   = blue,         % color of external links
  linkcolor  = darkblue,     % color of internal links
  citecolor  = darkgreen,    % color of links to bibliography
  filecolor  = darkred       % color of file links
}

\newtheorem{theorem}{Theorem}
\newtheorem{lemma}{Lemma}

\theoremstyle{definition}
\newtheorem{algorithm}{Algorithm}
\newtheorem{definition}{Definition}
\newtheorem{proposition}{Proposition}
\newcounter{factc}
\newtheorem{fact}[factc]{Fact}
\newtheorem{remark}{Remark}

\newcommand{\be}{\begin{eqnarray}}
\newcommand{\ee}{\end{eqnarray}}

\newcommand{\semigeq}{\succeq}

\newcommand{\QMA}{\mathsf{QMA}}

\DeclareMathOperator{\tr}{Tr}     % trace

\newcommand{\knote}[1]{\textcolor{red}{({\bf Kevin:} #1)}}
\newcommand{\onote}[1]{\textcolor{blue}{({\bf Ojas:} #1)}}

\renewcommand{\knote}[1]{}
\renewcommand{\onote}[1]{}

\newcommand{\Sh}{\mathcal{S}}
\newcommand{\psuedo}[1]{\widetilde{#1}}
\newcommand{\ketbra}[2]{\mathinner{|{#1}\rangle \langle{#2}|}}

\DeclareMathOperator{\Id}{\mathbb{I}}   % Identity
\DeclareMathOperator{\las}{\mathcal{L}} % Lasserre SDP
\DeclareMathOperator{\herm}{\mathcal{H}}   % Hermitian operators
\DeclareMathOperator{\D}{\mathcal{D}}   % density operators

\DeclareMathOperator{\C}{\mathbb{C}}
\DeclareMathOperator{\R}{\mathbb{R}}

\author{
  Ojas Parekh\orcidlink{0000-0003-2689-9264}\footnote{Sandia National Laboratories, {\it email:} odparek@sandia.gov} \,\,and  
  Kevin Thompson\orcidlink{0000-0001-5669-2200}\footnote{Sandia National Laboratories, {\it email:} kevthom@sandia.gov}
}

\title{An Optimal Product-State Approximation for 2-Local Quantum Hamiltonians with Positive Terms}

\begin{document}

\date{}
\maketitle

\begin{abstract}
We resolve the approximability of the maximum energy of the Quantum Max Cut (QMC) problem using product states.  A classical 0.498-approximation, using a basic semidefinite programming relaxation, is known for QMC, paralleling the celebrated 0.878-approximation for classical Max Cut.  For Max Cut, improving the 0.878-approximation is Unique-Games-hard (UG-hard), and one might expect that improving the 0.498-approximation is UG-hard for QMC.  In contrast, we give a classical 1/2-approximation for QMC that is unconditionally optimal, since simple examples exhibit a gap of 1/2 between the energies of an optimal product state and general quantum state.  Our result relies on a new nonlinear monogamy of entanglement inequality on a triangle that is derived from the second level of the quantum Lasserre hierarchy.  This inequality also applies to the quantum Heisenberg model, and our results generalize to instances of Max 2-Local Hamiltonian where each term is positive and has no 1-local parts.  Finally, we give further evidence that product states are essential for approximations of 2-Local Hamiltonian.
\end{abstract}

\tableofcontents
\section{Introduction}\label{sec:introduction}
The $\QMA$-hard $k$-Local Hamiltonian ($k$-LH) problem is a central problem in Quantum Complexity, playing a role akin to constraint satisfaction problems (CSPs) in classical algorithms and complexity theory.  Local Hamiltonian bridges a unique gap in that it has shaped quantum computational complexity theory while enjoying physical applications such as better understanding properties of materials.  Rigorous approximation algorithms\footnote{In our context, an \emph{$\alpha$-approximation algorithm} is a randomized classical polynomial-time algorithm producing a classical state or description of a quantum state; the output is guaranteed to have value at least $\alpha$ times the optimal.} and related hardness results remain relatively sparse, with a recent thrust to better understand approximability~\cite{G12, B16, B19, G19, H20, PT20, A20, PT21, A21, H21}.  The seminal Goemans-Williamson Max Cut approximation algorithm \cite{G95}, or other appropriate classical counterparts~\cite{B10,BRS11,R12}, have provided the inspiration for many of these works.  We expect that classical CSPs have much more to teach us about avenues for and barriers to approximating $k$-LH.

The classical Lasserre hierarchy of convex semidefinite programs (SDPs), dual to the Sum-of-Square hierarchy, has played a pivotal role in understanding the approximability of CSPs.  The Goemans-Williamson algorithm employs the \nth{1} level of the hierarchy and achieves the best-possible approximation guarantee for Max Cut under the Unique Games conjecture~\cite{K07}.  Raghavendra and Steurer show that this phenomenon holds for all constraint-satisfaction problems (CSPs), namely the basic level-1 SDP yields the best approximation algorithm for general instances~\cite{RS09}, and it is Unique-Games-hard to do better~\cite{R08}.  One may ask whether comparable results hold for $k$-LH.  Recent hardness and approximation results~\cite{H21,A20,PT21} provide evidence\footnote{Personal communication: A flaw was recently discovered in one of the proofs of \cite{H21}; the authors are attempting to fix it and believe their claims to be true.  We view their work as evidence in favor of their claimed hardness results.} that , unlike the classical CSP case, better approximations for $2$-LH are possible using the \nth{2} level of a quantum analogue of the Lasserre hierarchy introduced by Pironio, Navascu{\'e}s, and Acin~\cite{P10}.  A precise understanding of the ways in which the approximability of $2$-LH compares and contrasts with $2$-CSP remains an open problem.  We resolve this question in the case of using product states to approximate the Quantum Max Cut problem and more general $2$-LH instances.  

\paragraph{Quantum Max Cut.}
Max Cut has advanced the art of classical approximation algorithms over several decades, since it a simply formulated NP-hard $2$-CSP capturing essential properties of general CSPs. As quantum approximation algorithms are in their infancy, such an analogue for $2$-LH is highly desirable for advancing the field. The $\QMA$-hard Quantum Max Cut (QMC) problem has emerged to fill this role: (i) QMC is physically motivated and closely related to the well-known anti-ferromagnetic Heisenberg model~\cite{G19}; (ii) its approximability is conjectured to be that of the general positive $2$-LH problem\footnote{\emph{Positive $2$-LH} refers to instances of $2$-LH in which each term is positive semidefinite; this is the most natural generalization of classical CSPs.}~\cite{PT20} (iii) moreover, techniques for converting approximation algorithms for QMC to those for positive 2-LH are known~\cite{PT20}, (iv) it is considerably simpler to analyze than general $2$-LH (e.g.~\cite{PT20}); (v) it may be viewed as a generalization of Max Cut, which has enabled generalizing classical approaches for both approximation~\cite{G19,A20,PT21} and hardness~\cite{H21}.

QMC seeks to find the maximum eigenvalue of the Hamiltonian $\sum_{ij \in E} w_{ij}(\Id - X_iX_j - Y_iY_j - Z_iZ_j)/4$, where $X_i,Y_i,Z_i$ are the Paulis acting on qubit $i$ and $w_{ij} \geq 0$ are weights.  The well-studied anti-ferromagnetic quantum Heisenberg model is obtained by dropping the $\Id$ from each term: $\Id - X_iX_j - Y_iY_j - Z_iZ_j \leftrightarrow - X_iX_j - Y_iY_j - Z_iZ_j$.  Similarly, Max Cut expressed as a Hamiltonian is $\sum_{ij \in E} (\Id - Z_iZ_j)/2$, and the anti-ferromagnetic Ising model is obtained by the replacement: $\Id - Z_iZ_j \leftrightarrow -Z_iZ_j$.  Thus the relationship between QMC and Max Cut is analogous to that of the quantum Heisenberg and Ising models.  The energy shift does not affect the complexity of the finding the maximum eigenvalue; however, approximability is affected:  $\Omega(1/\log(n))$ for Ising/Heisenberg~\cite{C04,B19} versus $\Omega(1)$ for (Quantum) Max Cut~\cite{G95,G19}.  Further parallels between Max Cut and QMC appear in~\Cref{sec:MC-QMC-parallels}.

\paragraph{Extensions to positive $2$-Local Hamiltonian.}  We posit that approximations for QMC will result in approximations with matching performance gaurantees for the more general positive $2$-LH problem.  In some sense this is the most general version of $2$-LH that capture the spirit of classical $2$-CSPs (where each term is a positive function of its boolean variables).  A QMC term on qubits $i,j$ acts as a rank-$1$ projector onto the singlet state on $i,j$ (see \Cref{rem:QMC-term} for details).  Since any positive $2$-local term may be expressed as a positive combination of rank-$1$ projectors, the approximability of the $2$-LH problem with rank-$1$ captures that of general positive $2$-LH.  Note that each Max Cut term acts as a (diagonal, in the computational basis) rank-2 projector, and the approximability of the classical Max AND problem captures that of general $2$-CSP.   

While it may seem that QMC is an isolated problem, it represents a maximially entangled case of positive $2$-LH, and there is evidence~\cite{PT20} that it is hardest instance of $2$-LH to approximate. The present work and~\cite{PT20} give techniques for turning QMC approximation algorithms into those for positive 
$2$-LH (although approximation guarantees for the latter are established numerically, as it is harder to analyze).  Hence we believe that QMC simultaneously serves a role akin to both Max Cut and Max AND in the quantum case.

\paragraph{Main result.}  Gharibian and Parekh~\cite{G19} introduced QMC and gave a $0.498$-approximation\footnote{We use 0.498 and analogous numerical estimates of approximation factors as shorthand for the exact approximation factor.} using product states, where a $\frac{1}{2}$-approximation is the best possible in the product state regime due to the gap between the optimal quantum and product-state values on a single edge. \cite{G19} employs the \nth{1} level of the quantum Lasserre hierarchy ($\las_k$ will denote level $k$), and Hwang, Neeman, Parekh, Thompson, and Wright~\cite{H21} provide evidence that $0.498$ is the integrality gap\footnote{In our context, the \emph{integrality gap} of an SDP relaxation is the minimum of the ratio of the optimal and SDP values over all instances.} of $\las_1$ for QMC. They also conjecture a tight related hardness result: a $(0.956+\varepsilon)$-approximation of the optimal product state for QMC is Unique-Games-hard, which would match the best known approximation guarantee for this problem~\cite{B10}.  In light of this evidence, one might anticipate hardness of a $(0.498+\varepsilon)$-approximation for QMC using product states.  In contrast we show:
\begin{theorem}\label{thm:PROD}
Quantum Max Cut admits a $\frac{1}{2}$-approximation, with respect to $\las_2$, outputting a product state.
\end{theorem}
This resolves the approximability of QMC using product states, and in the process yields a classical CSP-like problem\footnote{Namely, approximating QMC with respect to a classically computable upper bound ($\las_2$ in this case) and producing a classical output (a product state).} whose approximability does not appear to be captured by $\las_1$ (a departure from the behavior of traditional CSPs).  Our result shows that approximating Quantum Max Cut via product states is one of limited generalizations of classical CSPs for which the natural SDP is likely to be suboptimal.  Another implication of our result is that even though it is likely NP-hard\footnote{Personal communication: John Wright has a proof sketch of NP-hardness.} to find the optimal product-state solution to QMC, this does not prevent us from obtaining the optimal worst-case approximation for QMC using product states.  

\begin{comment}
Moving beyond product states, Anshu, Gosset, and Morenz Korol~\cite{A20} delivered a 0.531-aproximation by using tensor products of 1- and 2-qubit states.
\begin{theorem}[from \cite{PT-ICALP}]\label{thm:ICALP}
Quantum Max Cut admits a 0.533-approximation, with respect to $\las_2$, outputting a tensor product of 1- and 2-qubit states.
\end{theorem}
The slightly improved approximation factor is not the main contribution of our work.  Rather, we develop a framework for employing $\las_2$ in approximating $2$-LH problems that we expect will influence future work.  
%This framework puts the design of quantum approximation algorithms on par with classical ones, enabling use of well-developed classical techniques.  
\end{comment}

Our approximation algorithm for QMC may be applied to the positive $2$-LH problem without $1$-local terms.  The analysis for the latter is more difficult, and the approximation factor is verified numerically. Additionally, we resolve an open question in quantum approximation.  Gharibian and Kempe~\cite{G12} show that for any instance $H$ of positive $2$-LH, there is a product state giving a $\frac{1}{2}$-approximation; more generally, they demonstrate existence of a $(1/2)^{k-1}$-approximate product state for positive $k$-LH.  Their result uses an optimal quantum state to construct a $\frac{1}{2}$-approximate product state.  Since one cannot assume access to an optimal state, their result is non-constructive.  Our algorithm may be viewed as a constructive version that, given $H$, produces a $\frac{1}{2}$-approximate product state (with respect to the $\las_2$ upper bound, so the approximation factor could be significantly better than $\frac{1}{2}$ for special cases).  

We note that 0.533 is the currently best-known approximation for QMC, using a modest generalization of product states~\cite{PT21,A20} and that approximations beyond $\frac{1}{2}$ using product states are possible on dense graphs, which is further discussed below.

\paragraph{New triangle inequality for the Heisenberg model.} Our algorithm may be viewed as an adaptation of an existing algorithm~\cite{PT21} to the product-state case.  However, proving that the resulting algorithm is optimal requires more precise bounds and analysis.  For this, we prove a new nonlinear triangle inequality on the energies of each term of the quantum Heisenberg model on a triangle.  
\begin{lemma}\label{lem:informal-triangle}
Let $\rho$ be an $n$-qubit density matrix.  Define $s_{ij} := \tr[S_{ij} \rho]$, where $S_{ij} = \frac{\Id + X_iX_j + Y_iY_j + Z_iZ_j}{2}$ is the SWAP operator on qubits $i,j$.  Then we have:
\begin{gather*}
0 \leq s_{ij}+s_{ik}+s_{jk} \leq 3,\\
s_{ij}^2 + s_{ik}^2 + s_{jk}^2 + 2(s_{ij}+s_{ik}+s_{jk}) - 2(s_{ij}s_{ik}+s_{ij}s_{jk}+s_{ik}s_{jk}) \leq 3.
\end{gather*}
\end{lemma}
Due to the close connection of Heisenberg terms, QMC terms, and the SWAP terms above, simple variants of \Cref{lem:informal-triangle} apply to all three cases.  The first linear inequality is the traditional monogamy of entanglement inequality, related to the Lieb-Mattis bound~\cite{L62}, that has been critically used in approximating QMC using tensor products of 1- and 2-qubit states~\cite{A20,PT21}.  The second nonlinear inequality is what we must employ for our more precise treatment of product states. We expect that our nonlinear inequality may be generalized to stars rather than triangles, as is the case for the first linear inequality~\cite{A20,PT21}; however, our results do not require such a generalization.  

We in fact prove a stronger version of the above result, demonstrating that \Cref{lem:informal-triangle} holds not just for values arising from a true quantum density matrix, but also for those derived from a relaxed (and not necessarily positive) quantum pseudo-density arising from $\las_2$ (\Cref{def:psuedo-density}).  To the best of our knowledge, the above lemma was previously unknown, despite decades of work on the Heisenberg model, and suggests the efficacy of using the quantum Lasserre hierarchy to derive new monogamy of entanglement constraints.  We expect for our work to initiate further inquires along these lines.

\paragraph{Why consider product states?} 
Since generic quantum states are highly entangled it is natural to guess that good approximation algorithms necessitate algorithms which produce entangled states.  We argue that the state-of-the-art in approximating QMC and $2$-LH has not moved beyond product states, but instead focuses on judiciously combining product states with entangled ansatzes. This was heralded by results of Brand{\~a}o and Harrow~\cite{B16}, demonstrating that, for Hamiltonians exhibiting regularity in their interaction graph, product states provide arbitrarily close to optimal approximations as a function of the vertex degree.  Hence, for an approximation algorithm to offer a good unconditional approximation guarantee, it must address the case of dense or high-degree graphs, where product states are known to give arbitrarily good approximations~\cite{G12, B16, A20}.  Our work reinforces this notion by demonstrating that the bottlenecks to good product-state approximations are low-degree subgraphs.  

We give a concrete demonstration of the above by quantifying the performance of approximation via product states for unweighted graphs with minimum degree at least $d$.  A related result by Anshu, Gosset, and Morenz is a $(\frac{4}{7}-O(|E|^{-1}))$-approximation for QMC on connected and unweighted graphs~\cite{A20}. We show that a slight modification of the standard product-state rounding algorithm of~\cite{G19} already outperforms the best-known general QMC approximation, using an \emph{entangled} ansatz, on unweighted graphs with minimum degree $3$.  Thus, a high-level message of our work is that a precise understanding of product states seems crucial for any kind of approximation of $2$-LH.

\paragraph{Open questions.} Does combining product-state rounding with specialized ansatzes for low-degree subgraphs give rise to best-possible approximations for QMC and positive $2$-LH?  Can our techniques be extended to give a $(1/2)^{k-1}$-approximation for positive $k$-LH using product states? Do broader classes of (new) monogamy of entanglement bounds arise from $\las_k$ for small $k$?

\paragraph{Subsequent Work.  }  After the preprint of this work many authors improved on the best known approximation factors for Quantum Max Cut \cite{lee2024improved, apte2025improved, huber2024second, gribling2025improved, lee2022optimizing}, with the current state of the art for general graphs being $0.611$ and the best conjectured achievable approximation factor being $0.625$\cite{apte2025conjectured}.  Additional new directions include the study of the complexity of Quantum Max Cut and more general Hamiltonians \cite{kallaugher_et_al:LIPIcs.ITCS.2025.63, piddock2025quantum}, the study of approximation algorithms for generalizations of Quantum Max Cut to higher dimensions \cite{jorquera2024monogamy, carlson2023approximation}, the study of approximation algorithms for the EPR Hamiltonian \cite{ju2025improved, tao2025refined, apte20250} as well as approximating constrained quantum problems \cite{parekh2024constrained, culf2024approximation}.  It is worth emphasizing that none of these subsequent results achieve (or can achieve) a better approximation factor than the current work for approximating Quantum Max Cut via product states as our result is optimal.

\section{Preliminaries}
\label{sec:preliminaries}

We use subscripts for indices of vectors and matrices; however, when enhancing readability, we will also use the notation $v(\cdot)$ and $M(\cdot,\cdot)$ to refer to entries of vectors or matrices, respectively.  We will use $S^{d-1}$ to refer to the set of unit vectors in $d$ dimensions: $S^{d-1}=\{x\in \mathbb{R}^d: x^T x=1\}$.  We will make use of the concept of a {\it graph}.  A graph is a tuple of vertices and edges, $G=(V, E)$, where E is a set of unordered pairs from $V\times V$.  A generic edge, for instance, will be denoted $ij\in E$.  In this paper the vertex set will always be $[n]=\{1, 2, ..., n\}$.  For a subset of edges $S\subseteq E$ we use the notation $\delta(S)$ to indicate the {\it edge boundary}: $\delta(S)$ is the set of edges $ij \in E$ such that there exists $k$ with $ik\in S$ or $jk\in S$, but $ij \notin S$. Similarly, the \emph{neighbors} of a vertex $i$ will be denoted $N(i)$: $N(i)=\{j\in V: ij \in E\}$.  A \emph{matching} is a set of edges such that no two share an endpoint.  A \emph{star graph} is a graph where all vertices are connected to a central node.  Formally, a star graph with $d$ leaves is a graph $G=(V, E)$ where $V=\{0, 1, ..., d\}$ and $E=\{0i:i\in V,\,\,i \neq 0\}$.  We will refer to a graph $G$ with minimum degree at least $d$ as a min-degree-$d$ graph.

The vertices of a graph can always be thought of as qubits in this paper, hence we reserve $n$ for the number of qubits.  We use $\herm_n$ to denote the set of Hermitian matrices on $n$ qubits: $\herm_n=\{O\in \mathbb{C}^{2^n \times 2^n } : O^\dagger =O \}$.  The notation $\D_n$ will represent the set of density matrices on $n$ qubits: $\D_n=\{\rho \in \mathbb{C}^{2^n \times 2^n }: \rho=\rho^\dagger,\,\, \rho\semigeq 0, \,\, \tr[\rho]=1 \}$.  We use the notation $\mathcal{U}_n$ for the set of unitary matrices on $n$ qubits: $\mathcal{U}_n=\{U\in \mathbb{C}^{2^n \times 2^n}: U^\dagger U=\mathbb{I}\}$.

The Pauli matrices are defined as:
\begin{equation*}
\label{eq:paulis}
 \mathbb{I}=\begin{bmatrix}
1 & 0 \\
0 & 1
\end{bmatrix},
\,\,\,\,\,\,
X=\begin{bmatrix}
0 & 1 \\
1 & 0
\end{bmatrix},
\,\,\,\,\,\,
Y=\begin{bmatrix}
0 & -i \\
i & 0
\end{bmatrix}, \,\text{and}
\,\,\,\,\,\,
Z=\begin{bmatrix}
1 & 0 \\
0 & -1
\end{bmatrix}.
\end{equation*}
\noindent Subscripts indicate quantum subsystems among $n$ qubits, as in \cite{PT21}.  For instance, the notation $\sigma_i$ is used to denote a Pauli matrix $\sigma \in \{X,Y,Z\}$ acting on qubit $i$, i.e.\ $\sigma_i := \mathbb{I} \otimes \mathbb{I} \otimes \ldots \otimes \sigma \otimes \ldots \otimes \mathbb{I} \in \mathbb{C}^{2^n \times 2^n}$, where the $\sigma$ occurs at position $i$.  We will frequently consider polynomials in the Pauli operators. 

\begin{definition}[Pauli polynomial] A \emph{term} or \emph{monomial} $\tau$ is a tensor product of Pauli operators, $\sigma_1 \otimes \ldots \otimes \sigma_n$ with $\sigma_i \in \{\Id,X,Y,Z\}$.  The \emph{degree} of $\tau$ is the number of $\sigma_i$ not equal to $\Id$.  We let $\mathcal{P}_n(k)$ be the set of Pauli terms of degree at most $k$ on $n$ qubits.  A \emph{Pauli polynomial} $P$, unless otherwise stated, is a real linear combination of terms; its degree is the maximal degree over its terms.  Pauli polynomials of degree $l$ correspond to $l$-local Hamiltonians, and we interchangeably refer to a Pauli polynomial and the Hermitian operator it represents.
\end{definition}

The algorithms we discuss will produce product states, which are convex combinations of tensor products of single-qubit density matrices.
\begin{definition}[Product state]
A \emph{product state} is $\rho \in \D_n$ of the form $\rho = \sum_l \mu_l \left(\bigotimes_{i \leq n} \rho^l_i\right)$, with $\sum_l \mu_l = 1$ and $\mu_l \geq 0$ for all $l$.  We generally assume $\rho$ has polynomial-size description in $n$.
\end{definition}

Lastly we note that the current result makes use of many observations from Parekh and Thompson~\cite{PT20}. We have placed clean statements of some of their lemmas we use in the Appendix (\Cref{sec:lems_from_beating_random}) for convenience.  

\subsection{Problem Statement}
We consider specializations of the strictly quadratic $2$-Local Hamiltonian problem.
\begin{definition}[Strictly quadratic $2$-Local Hamiltonian]
\label{def:strictly-quadratic-2-LH}
A \emph{strictly quadratic} $2$-local Hamiltonian is one of the form
\begin{equation*}
H = \sum_{ij \in E}w_{ij}H_{ij} =  \sum_{ij \in E}w_{ij}\left(c_{ij}(\Id)\cdot \Id + \sum_{\tau,\sigma \in \{X,Y,Z\}} c_{ij}(\tau_i\sigma_j)\cdot \tau_i\sigma_j\right),
\end{equation*}
where we assume $w_{ij} > 0$, $c_{ij}(\cdot) \in \R$, and $H_{ij} \succeq 0$ for each edge $ij \in E$.
\end{definition}
The key feature of a strictly quadratic instance is the absence of $1$-local Pauli terms.
\begin{definition}[Quantum Max Cut]\label{def:QMC-def}
\emph{Quantum Max Cut} (QMC) is an instance of strictly quadratic 2-local Hamiltonian where $H_{ij} = \frac{\Id - X_iX_j - Y_iY_j - Z_iZ_j}{4}$ for all $ij \in E$.
\end{definition}

\begin{remark}[Quantum Max Cut local term]\label{rem:QMC-term}
Each QMC local term is the density matrix of a projector onto a rank-$1$ pure state, normally referred to as the \emph{singlet}: $(\Id- X\otimes X - Y\otimes Y - Z\otimes Z)/4=(\ket{01}-\ket{10})(\bra{01}-\bra{10})/2$.  The diagonal part of this projector is $(\ketbra{01}{01} + \ketbra{10}{10})/2$, a scaled version of the classical Max Cut Hamiltonian, that earns maximal value on an edge by putting ``anti-aligned'' (i.e., of different signs) $\{\pm1\}$ Boolean values on its endpoints.  For QMC, a quantum state earns maximal value on an edge by assigning an anti-aligned superposition of the anti-aligned basis states on the edge -- i.e., with a singlet.   
\end{remark}

Given such a Hamiltonian, the goal of the \emph{Local Hamiltonian Problem} is to find the largest eigenvalue $\lambda_{\max}(H)$.  Naturally the main constraint in the quantum case is that a quantum state achieving value $\lambda_{\max}(H)$ may not have an efficient description.  Since optimizing over this domain is intractable in general, we seek an approximation to $\lambda_{\max}(H) \geq 0$ by optimizing over classes of states with efficient descriptions.

\begin{definition}[$\alpha$-Approximation algorithm]
An algorithm is an \emph{$\alpha$-approximation}, for some $\alpha \in (0,1]$, if given any instance $H$ of $2$-Local Hamiltonian within some class, it produces a (description of a) state $\rho \in \D_n$ with
\begin{equation*}
    \tr[H\rho] \geq \alpha \lambda_{\max}(H).
\end{equation*}
An $\alpha$-approximation runs in time polynomial in the size of its input.
\end{definition}
Our goal is to produce randomized classical approximation algorithms, that output product states, for strictly quadratic instances of $2$-Local Hamiltonian.

\subsection{Quantum Lasserre Hierarchy}
\label{sec:intro-quantum-lasserre-hierarchy}
The quantum Lasserre hierarchy has become a standard ingredient in the design and analysis of quantum approximation algorithms.  The noncommutative (quantum) Lasserre hierarchy, for a fixed Hamiltonian $H$, is a sequence of semidefinite programs (SDPs), of increasing size, meant to approximate moments $\tr[\rho\, \tau]$ of a quantum state $\rho$ on $n$ qubits with respect to Pauli terms $\tau$ of some fixed degree (degree-$2$ terms in our case since $H$ is $2$-local).  This hierarchy was introduced by Navascu{\'e}s, Pironio, and Acin~\cite{NPA08,P10} and related to work by Doherty, Liang, Toner, and Wehner~\cite{D08}.  It was first used in approximating local Hamiltonian problems by Brand{\~a}o and Harrow~\cite{B16}.  Level $k$ of the hierarchy, denoted $\las_k$, is an SDP over matrices of size $O(n^k)$ by $O(n^k)$ that gives a relaxation of moments over Pauli terms of degree at most $2k$.  For an $l$-local Hamiltonian $H \in \herm_n$, the SDP $\las_k$ may be used to obtain an upper bound on $\lambda_{\max}(H)$, for $k \geq \lceil l/2 \rceil$.  The value of this upper bound is monotonically decreasing in $k$, and when $k=n$ the bound is tight.  

\begin{definition}[Psuedo-density]\label{def:psuedo-density}
A \emph{$k$-positive psuedo-density} operator, $\psuedo{\rho} \in \herm_n$ is defined by:
\begin{enumerate}%[label=(\roman*)]
    \item  $\tr[\psuedo{\rho}]=1$, and
    \item  \label{item:psuedo-positivity}
    $\tr[\psuedo{\rho} P^2] \geq 0$, for all Pauli polynomials $P$ of degree at most $k$.
\end{enumerate}
The set of $k$-positive pseudo-density operators on $n$ qubits is denoted $\widetilde{\D}_n(k)$.  
%We note condition (ii) is equivalent to $\tr[\psuedo{\rho} Q] \geq 0$ for every positive $Q \in \herm_n$ of degree at most $2k$, since any such $Q$ may be expressed as a sum of squares of Hermitian operators of degree at most $2k$~\cite{}. 
\end{definition}
Level $k$ of the quantum Lasserre hierarchy yields an upper bound $\nu_k(H)$ on $\lambda_{\max}(H)$ by finding an optimal $k$-positive pseudo-density matrix, which amounts to solving an SDP:
\begin{equation} \label{prob:lasserre-psuedo-density}
\nu_k(H) := \max_{\psuedo{\rho} \in \widetilde{\D}_n(k)} \tr[H \psuedo{\rho}] \geq \lambda_{\max}(H).
\end{equation}
If we had $\psuedo{\rho} \succeq 0$, the inequality would be tight.  Indeed, \Cref{item:psuedo-positivity} of \Cref{def:psuedo-density} is a relaxation of $\psuedo{\rho} \succeq 0$, and when $k=n$ we are guaranteed $\psuedo{\rho} \succeq 0$.  

\begin{lemma}\label{lem:n-positivity}
$\psuedo{\D}_n(n) = \D_n$
\end{lemma}
\begin{proof}
Let $\psuedo{\rho} \in \psuedo{\D}_n(n)$.  Now for any unit $\ket{\psi} \in (\C^2)^{\otimes n}$, set $P = \ketbra{\psi}{\psi}$ in \Cref{item:psuedo-positivity} to yield $\braket{\psi|\psuedo{\rho}|\psi} = \tr[\psuedo{\rho}P] = \tr[\psuedo{\rho}P^2] \geq 0$. 
\end{proof}
To summarize, $\psuedo{\rho} \in \widetilde{\D}_n(k) \supset \widetilde{\D}_n(k+1) \supset \cdots \supset \widetilde{\D}_n(n) = \D_n$ is a relaxation of positivity.

Although pseudo-densities offer an elegant perspective into the quantum Lasserre hierarchy, our treatment will not explicitly refer to them.  Instead we will think of $\las_k$ as providing a relaxed value for each Pauli term in $H$.

\begin{definition}[Relaxed Hamiltonian objective] \label{def:relaxed-objective}
Let $\psuedo{\rho}$ be an optimal pseudo-density solution to $\las_k$, achieving value $\nu_k(H)$ (see \Cref{prob:lasserre-psuedo-density}).  For each Pauli term $\phi$ of $H$, let the \emph{relaxed value} of $\phi$ be
\begin{equation*}
\braket{\phi} := \tr[\psuedo{\rho} \phi].    
\end{equation*}
Based on \Cref{def:strictly-quadratic-2-LH}, the \emph{relaxed objective value} from $\las_k$ is then:
\begin{equation*}
\nu_k(H) = \sum_{ij \in E}w_{ij}\left(c_{ij}(\Id) + \sum_{\tau,\sigma \in \{X,Y,Z\}} c_{ij}(\tau_i\sigma_j)\braket{\tau_i\sigma_j}\right).
\end{equation*}
We will find it convenient to take
\begin{align*}
\mu_{ij} &:= c_{ij}(\Id) + \sum_{\tau,\sigma \in \{X,Y,Z\}} c_{ij}(\tau_i\sigma_j)\braket{\tau_i\sigma_j},\\
v_{ij}&:= \frac{4}{3} \sum_{\tau,\sigma \in \{X,Y,Z\}} c_{ij}(\tau_i\sigma_j)\braket{\tau_i\sigma_j}
\end{align*}
so that $\nu_k(H) = \sum_{ij} w_{ij} \mu_{ij}$.  
\end{definition}
\begin{remark}[Quantum Max Cut edge values]\label{rem:qmc_values}
For Quantum Max Cut, $\mu_{ij}$ and $v_{ij}$ specialize to
\begin{align*}
\mu_{ij} &:= \frac{1 - \braket{X_iX_j} - \braket{Y_iY_j} - \braket{Z_iZ_j}}{4},\\
v_{ij}&:= -\frac{\braket{X_iX_j} + \braket{Y_iY_j} + \braket{Z_iZ_j}}{3},
\end{align*}
where we will be able to assume $\mu_{ij} \in [0,1]$ and $v_{ij} \in [-\frac{1}{3},1]$ (\Cref{lem:QMC-values-ranges}).  The latter is interpreted as the average value earned on each of the three traceless terms.  These values will play a similar role for more general $2$-local Hamiltonian problems where each local term acts as a rank-$1$ projector on the corresponding 2-qubit subspace.
\end{remark}

We will derive approximation guarantees based on the expected performance of an approximation algorithm on an edge.
\begin{proposition} \label{prop:approx}
Let $\rho \in \D_n$ be a (description of a) solution generated by an approximation algorithm for an instance $H = \sum_{ij \in E} w_{ij} H_{ij}$ of a $2$-local Hamiltonian problem.  If for each $ij \in E$:
\begin{equation*}
Tr[H_{ij}\rho] \geq \alpha \mu_{ij},
\end{equation*}
for some $\alpha \in (0,1]$, then $\rho$ constitutes an $\alpha$-approximation of $\lambda_{\max}(H)$.
\end{proposition}
\begin{proof}We have
$\tr[H\rho] = \sum_{ij \in E} w_{ij} \tr[H_{ij}\rho] \geq \alpha \sum_{ij \in E} w_{ij}\mu_{ij} = \alpha \nu_k(H) \geq \alpha \lambda_{\max}(H),$ where the first inequality follows from the hypothesis and the second by \Cref{prob:lasserre-psuedo-density}.
\end{proof}

\subsection{Vector Solutions from the Quantum Lasserre Hierarchy}

We have discussed the relaxed values provided by $\las_k$, but our rounding algorithms will need more than just these values.  In particular we think of the output of solving $\las_k$ as a collection of unit vectors $\ket{\tau} \in \C^d$, for each Pauli term $\tau$ of degree at most $k$; we may assume $d = O(n^k)$.  

To motivate this notion of a solution, consider a true pure quantum state $\rho = \ketbra{\psi}{\psi}$, and let $\ket{\tau} := \tau \ket{\psi}$, for each $\tau$.  Then we would define the value as $\braket{\tau} := \tr[\rho \tau]$.  For any Pauli terms $\tau$ and $\sigma$ we would have $\braket{\tau | \sigma} = \braket{\psi | \tau \sigma | \psi} = \braket{\tau \sigma}$.  It turns out we may view $\las_k$ as providing vectors with this property, for any $\tau, \sigma$ of degree at most $k$.  This gives a different characterization of $\nu_k(H)$, as the following vector program, equivalent to an SDP.  A more complete and formal treatment, that also demonstrates equivalence to the pseudo-density perspective, is given in \Cref{sec:lasserre}.
\begin{definition}[Quantum Lasserre hierarchy] \label{def:intro-lasserre}
We are given as input $H = \sum_{\phi \in \mathcal{P}_n(2k)} c(\phi) \phi$, with $H \in \herm_n$.  Level $k$ of the Quantum Lasserre hierarchy, denoted $\las_k$, is defined by the following vector program:
\begin{alignat*}{2}
\nonumber 
    \nu_k(H) := \max \sum_{\phi \in \mathcal{P}_n(2k)} & c(\phi) \braket{\phi}\\
    s.t. \qquad
    %\label{cons:lasserre-identity}
    %\braket{\Id | \Id} &= 1 \\
    %\label{cons:lasserre-unit}
    \braket{\tau | \tau} &= 1
        \quad && \forall \tau \in \mathcal{P}_n(k) \\
    %\label{cons:lasserre}
    \braket{\tau | \sigma} &= \braket{\tau\sigma} 
        \quad && \forall \tau,\sigma \in \mathcal{P}_n(k) \\
\nonumber
    \ket{\tau} &\in \C^d
        \quad && \forall \tau \in \mathcal{P}_n(k),
\end{alignat*}
for any integer $d \geq |\mathcal{P}_n(k)|$.
\end{definition}

\section{Algorithm and Analysis Overview}

We begin with a high-level motivation of product-state approximations for local Hamiltonian problems as well as our algorithm and analysis.

\subsection{Parallels between Max Cut and Quantum Max Cut}
\label{sec:MC-QMC-parallels}

We motivate both relaxations and approximations for QMC by drawing parallels to the classical Max Cut problem.  This will help place our main result in contrast with analogous classical results.  The SDP relaxation for Max Cut employed in the celebrated Goemans-Williamson algorithm~\cite{G95} is typically cast as the following vector program.  
\begin{definition}[Max Cut vector program]
\label{def:MC-relaxation}
For a nonnegatively weighted graph $G=(V,E)$ on $n$ vertices define:
\begin{alignat*}{2}
    \nu_{MC} := \max \sum_{ij \in E} & w_{ij} \mathrlap{\frac{1 -\braket{Z_i|Z_j}}{2}}\\
    s.t. \qquad
        \braket{Z_i|Z_i} &=1  \qquad && \forall i \in V\\
        \ket{Z_i} &\in \R^d
        && \forall i \in V.
\end{alignat*}
The value $\nu_{MC}$ is an upper bound on the weight of a maximum cut in $G$ and remains the same for any integer $d \geq n$.  We use bra-ket notation for vectors for analogy with quantum relaxations used for Quantum Max Cut.
\end{definition}

A solution of value $\nu_{MC} - \varepsilon$ may be obtained in time polynomial in the instance size and $\log(1/\varepsilon)$. For such a solution, $\{\ket{Z_i}\}_{i \in V}$ the Goemans-Williamson algorithm selects a random Gaussian vector $\ket{r}$ and assigns each vertex $i$ to a side of a cut based on the sign of $\braket{Z_i|r}$.  This gives a randomized $\alpha_{GW}$-approximation, where $\alpha_{GW} \approx 0.878$. Khot, Kindler, Mossel, and O'Donnell demonstrate that is Unique-Games-hard (UG-hard) to obtain an $(\alpha_{GW} + \varepsilon)$-approximation, for any $\varepsilon > 0$~\cite{K07}.  Raghavendra and Steurer show that this phenomenon holds for all constraint-satisfaction problems (CSPs), namely that an approximation matching the integrality gap of the basic level-1 SDP relaxation is possible~\cite{RS09}, and it is UG-hard to do better~\cite{R08}.  

%The vector program above corresponds to the first level of the Lasserre/Sum-of-Squares hierarchy.  Although it may be strengthened by considering higher levels, by~\cite{K07} this is unlikely to yield a better approximation algorithm for general instances.  

As local Hamiltonian problems are natural quantum analogues of CSPs, one might strive for quantum analogues of the above results.  For this one first needs SDP relaxations for local Hamiltonian problems.

\paragraph{Quantum Max Cut relaxation.} Gharibian and Parekh~\cite{G19} obtain a vector-program relaxation for Quantum Max Cut in the vein of that for Max Cut.
\begin{definition}[Quantum Max Cut vector program]
\label{def:QMC-relaxation}
For a nonnegatively weighted graph $G=(V,E)$ on $n$ vertices define:
\begin{alignat*}{2}
    \nu_{QMC} := \max \sum_{ij \in E} & w_{ij} \mathrlap{\frac{1 - \braket{X_i|X_j} - \braket{Y_i|Y_j} - \braket{Z_i|Z_j}}{4}}\\
    s.t. \qquad
        \braket{\tau_i|\tau_i} &=1  \qquad && \forall i\in V\text{ and }\tau \in \{X,Y,Z\}\\
        \braket{\tau_i | \sigma_i} &= 0 
        \qquad && \forall i\in V\text{ and }\tau,\sigma \in \{X,Y,Z\} : \tau \not= \sigma\\
    \ket{\tau_i} &\in \R^d
        && \forall i\in V\text{ and }\tau \in \{X,Y,Z\}.
\end{alignat*}
The value $\nu_{MC}$ is an upper bound on $\lambda_{\max}(\sum_{ij \in E} w_{ij}(\Id-X_iX_j-Y_iY_j-Z_iZ_j)/4)$ and remains the same for any integer $d \geq 3n$ (proven more generally in \Cref{sec:lasserre}).  In contrast to Max Cut, for each vertex $i$ corresponding to a qubit, we have three mutually orthogonal unit vectors, $\ket{X_i}$, $\ket{Y_i}$, and $\ket{Z_i}$.
\end{definition}

This vector program may be further simplified to closely resemble that of Max Cut.  

\begin{definition}[Simplified Quantum Max Cut vector program]
\label{def:SQMC-relaxation}
For a nonnegatively weighted graph $G=(V,E)$ on $n$ vertices define:
\begin{alignat*}{2}
    \nu_{SQMC} := \max \sum_{ij \in E} & w_{ij} \mathrlap{\frac{1 - 3\braket{W_i|W_j}}{4}}\\
    s.t. \qquad
        \braket{W_i|W_i} &=1  \qquad && \forall i\in V\\
    \ket{W_i} &\in \R^d
        && \forall i\in V.
\end{alignat*}
The value $\nu_{SQMC}$ remains the same for any $d \geq n$.  The value of edge $ij$ in this context is $v_{ij}=-\braket{W_i|W_j}$.
\end{definition}
Indeed, it follows from~\cite{G19} that the two QMC relaxations above are equivalent:
\begin{lemma}\label{lem:QMC_to_MC}$\nu_{QMC} = \nu_{SQMC}$.
\end{lemma}
\begin{proof}
First we show $\nu_{SQMC} \geq \nu_{QMC}$.
For an optimal solution $\{\ket{X_i},\ket{Y_i},\ket{Z_i}\}_{i \in V}$ to the vector program of \Cref{def:QMC-relaxation}, we set $\ket{W_i} := (\ket{X_i} \oplus \ket{Y_i} \oplus \ket{Z_i})/\sqrt{3}$.  We have $\braket{W_i|W_i} = 1$ and $\braket{W_i|W_j} = (\braket{X_i|X_j} + \braket{Y_i|Y_j} + \braket{Z_i|Z_j})/3$, so $\{\ket{W_i}\}_{i \in V}$ is a feasible solution of objective value $\nu_{QMC}$ for the simplified QMC vector program.  Note that the values of the edges do not change under this transformation.  

For $\nu_{QMC} \geq \nu_{SQMC}$, take an optimal solution $\{\ket{W_i}\}_{i \in V}$ to the simplified QMC vector program, and define: 
\begin{align*}
\ket{X_i} &:= \ket{W_i} \oplus \vec{0} \oplus \vec{0} \\
\ket{Y_i} &:= \vec{0} \oplus \ket{W_i} \oplus \vec{0} \\
\ket{Z_i} &:= \vec{0} \oplus \vec{0} \oplus \ket{W_i}. 
\end{align*}
\noindent We have that $\ket{X_i}$, $\ket{Y_i}$, and  $\ket{Z_i}$ are mutually orthogonal unit vectors with $3\braket{W_i|W_j} = \braket{X_i|X_j} + \braket{Y_i|Y_j} + \braket{Z_i|Z_j}$.  This gives us a feasible solution to the program of \Cref{def:QMC-relaxation} with objective $\nu_{SQMC}$.
\end{proof}
The simplified QMC and Max Cut vector programs differ only in their objective values: 
\begin{equation}\label{eq:MC-QMC-objectives}
\nu_{QMC} = \frac{3\nu_{MC}-1}{2}.
\end{equation}
A natural question is whether the similarity in relaxations leads to a similarity in approximability of Max Cut and Quantum Max Cut.

\paragraph{Approximating (Quantum) Max Cut.}  \Cref{alg:QMC} is an approximation algorithm for QMC that produces a product state as a solution.  Its approximation guarantee matches that of \cite{G19}, but it is adapted to use the simplified QMC relaxation from \Cref{def:SQMC-relaxation} and more closely resembles the Goemans-Williamson algorithm for Max Cut.
%\begin{mdframed}
%\begin{minipage}{0.40\textwidth}
\begin{algorithm}[Max Cut]\label{alg:MC}
%\mbox{}\\[1ex]
\emph{Input: }$\ket{Z_i} \in \R^d$ for each $i \in V$
\begin{enumerate}
    \item Pick a random $\ket{r} \sim \mathcal{N}(0,1)^d$
    \item Output $u_i := \text{Unit}(\braket{Z_i|r})\ \forall i$
\end{enumerate}
\end{algorithm}
%\end{minipage}
%\hfill
%\begin{minipage}{0.57\textwidth}
\begin{algorithm}[Quantum Max Cut]\label{alg:QMC}
%\mbox{}\\[1ex]
\emph{Input: }$\ket{W_i} \in \R^d$ for each $i \in V$
\begin{enumerate}
    \item Pick random $\ket{r_x},\ket{r_y},\ket{r_z} \sim \mathcal{N}(0,1)^d$
    \item Output $u_i := \text{Unit}(\braket{W_i|r_x},\braket{W_i|r_y},\braket{W_i|r_z})\ \forall i$
\end{enumerate}
\end{algorithm}
%\end{minipage}
%\end{mdframed}

We include the Goemans-Williamson algorithm above for comparison with \Cref{alg:QMC}, and $\text{Unit}(v)$ returns a unit vector in the direction of $v$.  While in the classical case $u_i \in \{-1,1\}$ assigns $i$ to a side of a cut, in the quantum case $u_i$ is taken as a Bloch vector for a pure state $\rho_i$ on qubit $i$.  The output of the algorithm is then the product state $\rho = \otimes_{i} \rho_i$.  The expected performance of \Cref{alg:MC} on edge $ij$ only depends on $\braket{Z_i|Z_j}$, while that of \Cref{alg:QMC} only depends on $\braket{W_i|W_j}$.  The precise expectation follows from work of Bri{\"e}t, de Oliveira Filho, and Vallentin:

\begin{lemma}[Lemma 2.1 from~\cite{BJ10}]\label{lem:Briet-et-al}
Let $u, v$ be unit vectors in $\R^d$ and let $Z \in \R^{k \times d}$ be a random matrix whose entries are distributed independently according to the standard normal distribution with mean 0 and variance 1. Then,
\begin{equation*}
\mathbb{E}\left[ \frac{Zu}{\Vert Zu \Vert} \cdot \frac{Zv}{\Vert Z v\Vert}\right] = F(k, u \cdot v),
\end{equation*}
where
\begin{equation}\label{eq:F_with_r}
F(k,t) := \frac{2}{k}\left(\frac{\Gamma((k+1)/2)}{\Gamma(k/2)}\right)^2 t \,_2F_1\left(\frac{1}{2},\frac{1}{2};\frac{k}{2} + 1;t^2\right),
\end{equation}
with $_2F_1$ the Gaussian hypergeometric function.
\end{lemma}

Take $u = \ket{Z_i}$, $v = \ket{Z_j}$, and $k=1$ for the Max Cut case, and take $u = \ket{W_i}$, $v = \ket{W_j}$, and $k=3$ for the QMC case.  The matrix $Z$ from the lemma has row $\ket{r}$ in the former case and rows $\ket{r_x}$, $\ket{r_y}$, and $\ket{r_z}$ in the latter.  From this one may see that the expected performance of \Cref{alg:MC,alg:QMC} on edge $ij$ is
\begin{equation*}
\frac{1-F(1,\braket{Z_i|Z_j})}{2}\text{ and } \frac{1-F(3,\braket{W_i|W_j})}{4},
\end{equation*}
for Max Cut and QMC, respectively. To bound the overall approximation ratio, it suffices to bound the worst-case approximation ratio on any edge, by the argument in \Cref{prop:approx}.  Hence the worst-case approximation ratios are
\begin{equation*}
\min_{t \in [-1,1]}\frac{1-F(1,t)}{1-t}\text{ and } \min_{t \in [-1,1/3]}\frac{1-F(3,t)}{1-3t},
\end{equation*}
for Max Cut and QMC, which evaluate to $\approx 0.878$ and $\approx 0.498$, respectively.  

\paragraph{Approximability of QMC using product states.} We summarize the parallels in the above approximations for Max Cut and QMC.
\begin{enumerate}%[label=(\roman*)]
\item The same relaxation is employed for both, differing only in objective value (\Cref{def:MC-relaxation,def:SQMC-relaxation}).
\item The rounding algorithm for QMC is a natural 3-dimensional analogue of that for Max Cut (\Cref{alg:MC,alg:QMC}), and the analysis of both algorithms boils down to the same hypergeometric function (evaluated at different values of $k$).
\item For Max Cut the integrality gap of the relaxation matches the approximation ratio of the algorithm, as proven by Feige and Schechtman~\cite{FS02}. For QMC this is conjectured to be the case as well~\cite{H21}.
\end{enumerate}
For Max Cut, it is Unique-Games-hard to obtain an constant-factor improvement in approximation ratio over \Cref{alg:MC}~\cite{K07}.  Since product states exhibit no entanglement and are classical in some sense, one might expect an analogous hardness result to hold for approximating QMC using product states.  As mentioned in \Cref{sec:introduction}, this is conjectured to be the case when one seeks to approximate the optimal product-state solution to QMC~\cite{H21} (as opposed to using a product state to approximate the optimal quantum solution).  Yet we show, perhaps surprisingly, that it is possible to improve upon \Cref{alg:QMC}, yielding an unconditionally optimal product-state approximation algorithm for QMC.

\subsection{Approximation Algorithms}
\subsubsection{Rounding Algorithm for Unweighted Graphs with Bounded Minimum Degree}
Our first algorithm is a relatively simple algorithm that performs well on unweighted instances of QMC of some minimum degree, $d$.  Our goal in including this algorithm is to demonstrate that product states perform well as the minimum degree $d$ increases.  This follows more generally from results of Brand{\~a}o and Harrow~\cite{B16}, but for QMC we are able to demonstrate this more directly.  For rounding we will simply apply the standard product state rounding technique (\Cref{alg:QMC}) to the second level of the Lasserre hierarchy.  A minor technical detail is that we must convert the vectors $\{\ket{X_i}, \ket{Y_i}, \ket{Z_i}\}$ to a single vector $\ket{W_i}$ to use \Cref{alg:QMC}, but this can be done without changing the edge values by the proof of \Cref{lem:QMC_to_MC}. Note that this rounding algorithm is nearly identical to \cite{G19}.  The main difference is that \Cref{alg:reg_graph} uses the tighter relaxation $\las_2$ rather than $\las_1$.  

\begin{mdframed}
\begin{algorithm}[Product-state approximation algorithm for Unweighted min-degree-$d$ Quantum Max Cut Instances]
\label{alg:reg_graph}
\mbox{}\\[1ex]
\emph{Input}: A QMC Hamiltonian $H=\sum_{ij \in E} H_{ij}$ (\Cref{def:QMC-def}) with $G=(V, E)$ a $d$-regular graph.
\begin{enumerate}
    \item  Solve the instance of level 2 of the quantum Lasserre hierarchy, $\las_2$ (\Cref{def:intro-lasserre}) for $H$.
    \item  Convert $\las_2$ vectors $\{(\ket{X_i}, \ket{Y_i}, \ket{Z_i}\}_{i\in V}$ into a single vector for each $i\in V$, $\{\ket{W_i}\}_{i\in V}$.  Execute \Cref{alg:QMC} on $\{\ket{W_i}\}_{i\in V}$ to find Bloch vectors for all vertices in $V$.  Denote the output Bloch vectors as: $\{(\theta_i^X, \theta_i^Y, \theta_i^Z)\}_{i \in V}$.
    \item  Output the state:
    \begin{equation}
        \rho =  \prod_{k \in V} \left( \frac{\mathbb{I}+\theta_k^X X_k + \theta_k^Y Y_k + \theta_k^Z Z_k}{2}\right)
    \end{equation}
    \end{enumerate}
\end{algorithm}
\end{mdframed}

%The analysis for the performance of this rounding algorithm proceeds by writing the Hamiltonian $H$ as a sum over ``star'' Hamiltonians as in~\cite{A20}.  By double counting each edge and dividing the overall Hamiltonian by a factor of $2$, we can decompose the Hamiltonian as a sum of star Hamiltonians where each star is rooted at one of the vertices:
%$$
%H=\frac{1}{2}\sum_{i\in V} h_i,\,\, \text{where} \,\, h_i=\sum_{j\in N(i)} H_{ij}.
%$$

%If $\las_2$ assigns a value $\langle h_v \rangle $ to the Pauli polynomial $h_v$, then we can demonstrate an approximation factor $\alpha$ by demonstrating that the random output of the rounding algorithm achieves $\mathbb{E}[\rho h_v] \geq \alpha h_v$.  THis fol

%Then, we can prove an approximation factor by demonstrating that the random solution from the rouding algorithm satisfies $\mathbb{E}(\rho h_v) \geq \alpha(d)$

%The analysis largely stems from a technical fact (LEMMA BLANK) concerning the effectiveness of the rounding algorithm of~\cite{G19} on \emph{star graphs} when the SDP is subject to the \emph{star bound}.  A \emph{star graph} is simply a graph in which every vertex is connected to a central node and there are no other edges in the graph.  The \emph{star bound} states that optimal solutions to $\las_2$ satisfy a certain inequality on subgraphs isomorphic to stars.  Essentially, for some fixed $i$, the \emph{values} (DEF BLANK) $\{v_{ij}\}_{j\in N(i)} $ respect a physically-motivated upper bound.   

\subsubsection{Rounding Algorithm for General Quantum Max Cut Instances}

We opt to present the rounding algorithm for QMC instances (\Cref{alg:main}) before the rounding algorithm for arbitrary PSD instances (\Cref{alg:generic} in \Cref{sec:positive-hamiltonians}).  \Cref{alg:generic} specializes to \Cref{alg:main}, however \Cref{alg:main} is easier to describe and we have analytic rather than numerical results for it.  In this direction let us assume we are given a QMC Hamiltonian $H=\sum_{ij\in E} w_{ij} H_{ij}$ with $w_{ij}\geq 0$ for all $ij\in E$.  

%We will assume further than each $w_{ij}\geq 0$ and each $H_{ij} \succeq 0$.  First, we solve the SDP corresponding to the second level of the quantum Lasserre hierarchy ($\las_2$) for our input Hamiltonian.  Let $\hat{\rho}$ be an optimal state for $H$.  Then, letting $\hat{\mu}_{ij} := \tr[H_{ij} \rho]$, we see that each optimal edge value $\hat{\mu}_{ij}$ is in $[0,1]$.  In our analysis, we will find it more convenient to work with the ``traceless'' value on each edge, $\hat{v}_{ij} := (4\hat{\mu}_{ij}-1)/3 \in [-1/3,1]$.  

First, we solve the SDP corresponding to the second level of the quantum Lasserre hierarchy ($\las_2$) for our input Hamiltonian.  As explained for \Cref{def:relaxed-objective}, from the solution to our $\las_2$ relaxation, we obtain relaxed values, $\mu_{ij}$ on each edge.  For $H_{ij}$ as described above, we will also have $\mu_{ij} \in [0,1]$ and define $v_{ij} := (4\mu_{ij}-1)/3 \in [-\frac{1}{3},1]$.  To avoid confusion, we will henceforth refer to the $\mu_{ij}$ as \emph{energies} that directly contribute to the objective function (see \Cref{prop:approx}), while we will call the $v_{ij}$ \emph{values}, corresponding to what is earned on the traceless part of each term.  The $\mu_{ij}$ represent relaxed energies (arising from a pseudo-density) and will generally not correspond to energies attainable by a quantum state.  %Because $\las_2$ is a relaxation, we are guaranteed that $\sum_{ij \in E} \mu_{ij} \geq \sum_{ij \in E}\hat{\mu}_{ij}$ (likewise $\sum_{ij \in E} v_{ij} \geq \sum_{ij \in E}\hat{v}_{ij}$).

The rounding algorithm picks the ``large'' edges $L$ as the set of edges $ij\in E$ with value $v_{i j}$ larger than some pre-defined threshold $\gamma$. Because the SDP solution obeys certain monogamy of entanglement constraints, adjacent large edges are prohibited.  Hence, the edges in $L$ form a matching and one may independently assign 
a product state earning $\frac{1}{2}$ the maximum energy (see \Cref{eq:planted-product-state}) on each such edge.  By fixing $2$-qubit states on each edge in $L$, we lose freedom on edges adjacent to $L$ and are relegated to random assignment on such edges.  For all remaining edges, standard Gharibian-Parekh product-state rounding is applied.

\begin{mdframed}
\begin{algorithm}[Product-state approximation algorithm for Quantum Max Cut]
\label{alg:main}
\mbox{}\\[1ex]
\emph{Input}: A QMC Hamiltonian $H=\sum_{ij\in E}w_{ij} H_{ij}$ (\Cref{def:QMC-def}), and a threshold $\gamma \in (-1/3,1]$.
\begin{enumerate}
    \item  Solve the instance of level 2 of the quantum Lasserre hierarchy, $\las_2$ (\Cref{def:intro-lasserre}) for $H$.
    \item  For each edge $ij$ calculate its value $v_{ij}$ from the $\las_2$ solution, according to \Cref{def:relaxed-objective}.
    \item  Let $L:=\{ij\in E: v_{ij} > \gamma\}$ be the set of ``large'' edges, which are guaranteed to be a matching (\Cref{rem:matching}).  Let $B$ be the set of vertices not adjacent to a large edge: $B=\{i\in [n]: ij \notin L \text{ for any }j\}$.
    \item \label{item:execute-product-state-algorithm} 
    Convert $\las_2$ vectors $\{(\ket{X_i}, \ket{Y_i}, \ket{Z_i}\}_{i\in B}$ into a single vector for each $i\in B$, $\{\ket{W_i}\}_{i\in B}$.  Execute \Cref{alg:QMC} on $\{\ket{W_i}\}_{i\in B}$ to find Bloch vectors for all vertices in $B$.  Denote the output Bloch vectors as: $\{(\theta_i^X, \theta_i^Y, \theta_i^Z)\}_{i \in B}$.
    \item  Output the state:
    \begin{equation}
        \rho = \prod_{ij \in L}\left( \frac{\mathbb{I}-Z_i Z_j}{4}\right) \prod_{k \in B} \left( \frac{\mathbb{I}+\theta_k^X X_k + \theta_k^Y Y_k + \theta_k^Z Z_k}{2}\right)
    \end{equation}
    \end{enumerate}
\end{algorithm}
\end{mdframed}

\paragraph{Obtaining a pure product state.} \Cref{alg:main} outputs a mixed product state $\rho$ that assigns:
\begin{equation}\label{eq:planted-product-state}
\frac{\Id-Z_iZ_j}{4} = \frac{1}{2}\ketbra{01}{01} + \frac{1}{2}\ketbra{10}{10},
\end{equation}
to the large edges, $L$, which we will show must form a matching.  To obtain a pure product state, we assign either $\ket{01}$ or $\ket{10}$ to each edge in $L$.  Our performance guarantee only requires independence of the random choice of assignment for pairs of edges in $L$, hence we may derandomize this choice efficiently using standard techniques~\cite{LW06}.
The algorithms of \Cref{item:execute-product-state-algorithm} return pure product states.

\subsection{Technical Overview}\label{sec:tech_overview}
Here we will present technical overviews for proving the performance of \Cref{alg:reg_graph} and \Cref{alg:main}.   

\paragraph{Min-degree-$d$ graphs.}
For \Cref{alg:reg_graph} we proceed by by writing the unweighted Hamiltonian $H$ as a sum over star Hamiltonians as in~\cite{A20}.  By double counting each edge and dividing the overall Hamiltonian by a factor of $2$, we can decompose the Hamiltonian as a sum of star Hamiltonians where each star is rooted at one of the vertices:
\begin{equation*}
H=\frac{1}{2}\sum_{i\in V} h_i,\,\, \text{where} \,\, h_i=\sum_{j\in N(i)} H_{ij}.
\end{equation*}
Since the graph has minimum degree $\geq d$ each $h_i$ is a star Hamiltonian with $\geq d$ many leaves.
Then, we can demonstrate an approximation factor $\alpha$ by demonstrating that $\mathbb{E}[\tr(\rho h_v)] \geq \alpha \langle h_v \rangle$ for all $v\in V$ where $\rho$ is the random output of \Cref{alg:reg_graph}.  We will show this for some constant $\alpha(d)$.  

For some fixed $h_i$ which corresponds to a star with $d$ leaves let $\{v_j\}_{j=1}^d$ be the values assigned by $\las_2$ (\Cref{rem:qmc_values}).  We will find a lower bound, $\alpha(d)$, for the approximation factor obtained by \Cref{alg:reg_graph} on $h_i$ and establish that $\alpha(d+1) \geq \alpha(d)$.  Since $\alpha(d+1) \geq \alpha(d)$, on stars of higher degree we obtain an even better approximation factor and the achieved approximation factor overall is $\geq \alpha(d)$.  We find $\alpha(d)$ by demonstrating that the values minimizing the approximation factor are the same and non-negative for each edge in the star without loss of generality.  This reduces the optimization to a single parameter which is then constrained by an inequality from $\las_2$ called the \emph{star bound} in previous work~\cite{A20, PT21}.  Let $F(3, v)$ be defined in \Cref{eq:F_with_r}.  The main result is:
\begin{theorem}\label{thm:approx_reg}
If $G$ is an unweighted min-degree-$d$ graph then \Cref{alg:reg_graph} is an $\alpha(d)$-approximation where
$$
\alpha(d)=\frac{1+F(3, \frac{1}{3}+\frac{2}{3d})}{2+\frac{2}{d}}
$$
\end{theorem}
We note that for $d=3$, we have $\alpha(3) > 0.557$, so that the above product-state approximation outperforms the best-known general $0.533$-approximation for QMC~\cite{PT21}.

\paragraph{General graphs.} Now we will describe the analysis for the algorithm for general graphs, \Cref{alg:main}.  We focus on the case of Quantum Max Cut (QMC) for this overview and the majority of the paper.  The more general case of strictly quadratic $2$-local Hamiltonian is treated in \Cref{sec:positive-hamiltonians}. Our main result establishes an optimal worst-case bound on the performance of \Cref{alg:main}.
\newtheorem*{thm:main}{Theorem \ref{thm:main}}
\begin{thm:main}There exists a $\gamma \in (-\frac{1}{3},1]$ such that \Cref{alg:main} is a $\frac{1}{2}$-approximation for QMC Hamiltonians.
\end{thm:main}

For QMC, Gharibian and Parekh~\cite{G19} apply an algorithm of Bri{\"e}t, de Oliveira Filho, and Vallentin~\cite{B10} to level 1 of the Quantum Lasserre Hierarchy, $\las_1$. They show that the expected performance of their algorithm on an edge $ij$ depends only on its value $v_{ij}$.  The worst case approximation guarantee is $\approx 0.498$ and occurs when an edge has $v_{ij} \approx 0.97$.  
%%% OP: reinstate once we fix the hardness results
%As previously mentioned, since the integrality gap of $\las_1$ for QMC precisely matches the performance guarantee of the Gharibian-Parekh algorithm, and improvement solely using $\las_1$ as an upper bound on $\lambda_{\max}(H)$ is impossible.

There is some hope in that if $v_{ij}$ is small, then the Ghariabian-Parekh algorithms has performance arbitrarily better than $\frac{1}{2}$.  However, it seems difficult to rule out the existence of large edges in a systematic way.  We draw inspiration from Parekh and Thompson~\cite{PT21}, who identified and treated large edges separately in designing an approximation for QMC employing entangled states.  Parekh and Thompson use an ansatz introduced by Anshu, Gosset, and Morenz~\cite{A20}, who place a singlet on select pairs of qubits combined with a product state on the remaining qubits.  Parekh and Thompson used a $\las_2$ solution to inform placement of singlets on pairs of qubits with suffciently large $v_{ij}$.  This ensured that the maximum possible energy was earned on such edges.  Instead, we place $(\Id-Z_iZ_j)/4$ on edges in $L$, which is a convex combination of pure product states. This ensures that $\frac{1}{2}$ the maximum energy is earned on such edges.  However, the catch, for both Parekh and Thompson's approach and ours, is that edges adjacent to large edges only earn $\frac{1}{4}$ of the maximum energy.  To overcome this, we would like to show that optimal energies, or in our case the relaxed $\las_2$ energies, cannot be large on both of an adjacent pair of edges.

\paragraph{Monogamy of entanglement bounds.}
Anshu, Gosset, and Morenz~\cite{A20} appealed to the well-known Lieb-Mattis monogamy of entanglement bound to bound the maximum energy a quantum state may earn on a star, while Parekh and Thompson~\cite{PT21} showed that this bound was implied by $\las_2$ and could be used for approximation algorithms based on $\las_2$ (i.e., was satisfied by $\las_2$ pseudo-densities).  In the case of a pair of adjacent edges, the star bound applies to the sum of the energies or values earned. Concretely, on two adjacent edges $ij$ and $jk$, the star bound implies $v_{ij} + v_{jk} \leq \frac{4}{3}$.  The latter is indeed sharp (at e.g., $v_{ij} = v_{jk} = \frac{2}{3}$).  This bound allows pairs of $v_{ij}$ and $v_{jk}$ such that the Gharibian-parekh algorithm is unable to achieve an overall $\frac{1}{2}$-approximation on these two edges.

We remedy this with a more refined analysis of $\las_2$ in order to establish a new nonlinear monogamy of entanglement bound more precisely capturing feasible energies on a triangle (see \Cref{lem:informal-triangle} and \Cref{sec:nonlinear-triangle-inequality}).  For strictly quadratic $2$-Local Hamiltonian, $\las_2$ produces a vector solution assigning up to 9 vectors (i.e., for pairs of nontrivial Pauli terms) per edge.  For QMC, Parekh and Thompson~\cite{PT21} were able to consider 3 vectors per edge.  Here, by performing a carefully chosen change of basis, we reduce this to a single vector per edge, enabling our new bound.  

From our triangle inequality, we deduce that there exists a $\gamma$ that satisfies the following.     
\setcounter{factc}{0}
\begin{fact}[by \Cref{lem:1/2-approx}]\label{fact:1}
    If $H$ is a QMC Hamiltonian, and if edge $ij$ has $v_{ij} \leq \gamma$ then product state rounding has approximation factor $\geq \frac{1}{2}$ on $ij$.
\end{fact}

\begin{fact}[by \Cref{lem:large-edge-bound}]\label{fact:2}
    If $H$ is a QMC Hamiltonian, and if an edge $ij$ has $v_{ij} > \gamma$ then any adjacent edge $ik$ has $v_{ik} < \frac{1}{3}$.
\end{fact}

\Cref{fact:1} states that if an edge has sufficiently small value then Gharibian-Parekh product rounding~\cite{G19} has performance $\geq \frac{1}{2}$ on that edge.  \Cref{fact:2} states that edges with sufficiently large value have sufficiently small values on neighboring edges.  This is useful because if $v_{ik} < \frac{1}{3}$, then independent random assignments on qubits $i$ and $k$ yield a $\frac{1}{2}$-approximation on $ik$.  Thus our analysis enables us to only focus on edges with $v_{ij} > \gamma$, and as previously mentioned, for such edges we explicitly assign a product state earning $\frac{1}{2}$ of the maximum energy. For reference, $\gamma=0.911$ in the context of this paper.  

\Cref{fact:2} can be seen as a monogamy of entanglement bound on pseudo-densities.  An edge which is highly entangled (has large value) forces adjacent edges to have small value (low entanglement). It is a tightening of the star bounds found in~\cite{A20} and~\cite{PT21}.  Indeed, these bounds only imply that when $v_{ij}>0.911$, then $v_{ik}<\frac{4}{3}-0.911 \approx 0.42$. 

%The implication for \Cref{alg:main} is that, for edges in or adjacent to $L$, we get a $1/2$-approximation by \Cref{fact:2}.  For the remaining edges, the product state produced in \Cref{item:execute-product-state-algorithm} is a $1/2$-approximation by \Cref{fact:1}.  This gives us an overall $1/2$-approximation by \Cref{prop:approx}.  We note that if there are no isolated edges in $H$, then our approximation is strictly better than $1/2$.

\paragraph{Positive Hamiltonians.} We generalize our QMC results to strictly quadratic $2$-Local Hamiltonian (\Cref{def:strictly-quadratic-2-LH}) by extending an approach of Parekh and Thompson~\cite{PT20} for approximating this problem via product states.  The main difficulty is the number of parameters involved in such more general $2$-Local Hamiltonian instances.  We prove that it suffices to consider 3 parameters and perform numerical search and experiments to justify a $\frac{1}{2}$-approximation for this case.

\paragraph{Organization.} In the \Cref{sec:lasserre}, we introduce the quantum Lasserre hierarchy and establish some fundamental properties.  Our treatment differs in perspective than previous approaches.  In \Cref{sec:product-state-approx-QMC} we establish our main results for the more simple case of Quantum Max Cut.  These results are generalized to the positive Hamiltonian case in \Cref{sec:positive-hamiltonians}. 
\section{Product-state Approximation for Quantum Max Cut}
\label{sec:product-state-approx-QMC}

Our goal in this section is to prove our main result.
\begin{theorem} \label{thm:main}
There exists a $\gamma \in (-1/3,1]$ such that \Cref{alg:main} is a $1/2$-approximation for QMC Hamiltonians.
\end{theorem}

We specialize the Lasserre hierarchy, $\las_k$ for Quantum Max Cut (QMC).  In this case we express the objective as 
\begin{equation*}
\nu_k := \sum_{ij \in E} w_{ij}\frac{1 - \braket{X_iX_j} - \braket{Y_iY_j} - \braket{Z_iZ_j}}{4} = \sum_{ij \in E} w_{ij}\frac{1 - \braket{X_i|X_j} - \braket{Y_i|Y_j} - \braket{Z_i|Z_j}}{4},
\end{equation*}
where the equality is by \Cref{def:term-value}.
Since $H = \sum_{ij \in E} w_{ij}(\Id - X_iX_j - Y_iY_j - Z_iZ_j)/4$ is the corresponding Hamiltonian, we have $\nu_k \geq \lambda_{\max}(H)$ by \Cref{lem:relaxation-bound}.  We will primarily focus on $\las_2$ but will need to refer to $\las_1$ in order to appeal to previous work on approximating QMC using product states.  For consistency with prior work we will consider the real version of the hierarchy defined in~\cite{G19}.  The real solution can be obtained from the complex one via \Cref{rem:real-lasserre} of \Cref{rem:lasserre}.  In the case of QMC, this relaxation specializes to that of \Cref{def:QMC-relaxation}.  The relaxation employs all vectors $\ket{\tau} \in \mathcal{P}_n(1)$ except $\ket{\Id}$.  The latter does not appear in the objective, and omitting it does not affect the value of the relaxation.  The symmetries of the Quantum Max Cut problem render it an excellent instance of $2$-Local Hamiltonian for developing new approximation techniques.  In particular, it suffices to focus on only a single scalar value, derived from an $\las_k$ solution, for each edge $ij$.  It will be helpful to consider some affine transformations of this value, outlined below.

\begin{definition}[Quantum Max Cut edge values from $\las_k$] \label{def:QMC-values}
Given a solution $\{\ket{\tau}\}_{\tau \in \mathcal{P}_k(n)}$ to $\las_k$, we define, for each edge $ij \in E$:
\begin{align*}
    \mu_{ij} &:= \frac{1-\braket{X_iX_j} - \braket{Y_iY_j} - \braket{Z_iZ_j}}{4},\\
    v_{ij} &:= -\frac{\braket{X_iX_j} + \braket{Y_iY_j} + \braket{Z_iZ_j}}{3} = \frac{4\mu_{ij}-1}{3},\\
    s_{ij} &:= \frac{1 + \braket{X_iX_j} + \braket{Y_iY_j} + \braket{Z_iZ_j}}{2} = \frac{1-3v_{ij}}{2}.
\end{align*}
Note that $\mu_{ij}$, $v_{ij}$ above are consistent with \Cref{def:relaxed-objective} when specialied to QMC as in \Cref{rem:qmc_values}.
\end{definition}

\begin{lemma} \label{lem:QMC-values-ranges}
For the values of \Cref{def:QMC-values} from $\las_k$, for $k \geq 2$ we have
\begin{equation*}
    \mu_{ij} \in [0,1],\ v_{ij} \in [-1/3,1],\text{ and } s_{ij} \in [-1,1].
\end{equation*}
\end{lemma}
\begin{proof}
Lemma 12 from \cite{PT21} shows that $\mu_{ij} \in [0,1]$ for $\las_2$, from which the other values follow.  For the $k>2$ case see \Cref{item:lasserre-(k-1)} of \Cref{rem:lasserre}.
\end{proof}

\subsection{A Threshold for 1/2-Approximability}
The primary value of interest will be $v_{ij}$, since the performance of the $\las_1$ product-state rounding algorithm of Gharibian and Parekh~\cite{G19} on each edge depends only on $v_{ij}$.  In particular, they show that for the product state $\rho$ their algorithm generates:
\begin{equation} \label{eq:GP-expectation}
\tr\left[\frac{\Id-X_iX_j-Y_iY_j-Z_iZ_j}{4} \rho\right] = \frac{1-F(3, -v_{ij})}{4} = \frac{1+F(3, v_{ij})}{4},
\end{equation}
where
\begin{equation} \label{def:F[v]}
    F(3, v) = \frac{8}{3 \pi } v \,_2 F_1\left(\frac{1}{2}, \frac{1}{2};\frac{5}{2}; v^2\right),
\end{equation}
as used by Bri{\"e}t, de Oliveira Filho, and Vallentin~\cite{BJ10} and Parekh and Gharibian (see \Cref{lem:Briet-et-al}). The last equality of \Cref{eq:GP-expectation} holds because $_2 F_1(1/2, 1/2;5/2; v^2) \geq 0$ for $v \in [-1,1]$. The value earned by $\las_k$ on edge $ij$ is $\mu_{ij} = (1+3v_{ij})/4$, hence if for each edge one has:
\begin{equation*}
1 + F(3, v_{ij}) \geq \alpha (1 + 3v_{ij}),
\end{equation*}
then one obtains an $\alpha$-approximation by \Cref{prop:approx}.  For an $\las_1$ solution, $v_{ij} \in [-1,1]$; however, if $v_{ij} \leq -1/3$ then $\las_1$ earns a nonpositive value on $ij$, while the approximation always earns a nonnegative value.  Thus the worst-case approximation on an edge is given by:
\begin{equation*}
\min_{v \in (-1/3,1]} \frac{1 + F(3, v)}{1 + 3v} \geq 0.498,    
\end{equation*}
which occurs at $v \approx 0.97$.  Our goal is a $1/2$-approximation for QMC using product states, and it is edges of large $v_{ij}$ value that impose a barrier to this using the algorithm of~\cite{G19}.  Specifically, there is a relatively large $\gamma$ such that if $v_{ij} \leq \gamma$ for every edge, then the Gharibian-Parekh algorithm is a $1/2$-approximation.  Such worst-case approximations over a parameter interval are typically derived by numerical search.  For our setting, we are able to give an analytic argument that entails checking the approximation ratio at a singe value of the parameter.
\begin{lemma}\label{lem:1/2-approx}
Define
\begin{equation*}
    h(v,\gamma) := \frac{8 (3 \gamma+1) \gamma^2 \, _2F_1\left(\frac{3}{2},\frac{3}{2};\frac{7}{2};\gamma^2\right)+40 \, _2F_1\left(\frac{1}{2},\frac{1}{2};\frac{5}{2};\gamma^2\right)-45 \pi }{15 \pi  (3 v+1)^2}.
\end{equation*}
For any $\gamma \in [1/3, 1)$\footnote{This lemma can be simpliy modified to handle $\gamma <1/3$, but it is sufficient for us to consider $\gamma \in [1/3, 1]$.} such that $h(v,\gamma) < 0$ for all $v \in (-1/3,\gamma]$ and $\frac{1 + F(3, \gamma)}{1 + 3\gamma} \geq 1/2$, we have 
\begin{equation*}
\min_{v \in (-1/3,\gamma]} \frac{1 + F(3, v)}{1 + 3v} \geq \frac{1}{2}.
\end{equation*}
\end{lemma}
\begin{proof}
Define 
\begin{equation}
    g(v):=\frac{1 + F(3, v)}{1 + 3v}.
\end{equation}

The derivative of $g(v)$ is easily found with the quotient rule, and the well known form of $d/dz \,_2 F_1 (a, b;c ;z)$ \cite{A48}.  It is:
\begin{equation}\label{eq:der_of_approx}
    h(v, v)=\frac{d}{dv} \left[ \frac{1 + F(3, v)}{1 + 3v}\right]=\frac{8 (3 v+1) v^2 \, _2F_1\left(\frac{3}{2},\frac{3}{2};\frac{7}{2};v^2\right)+40 \, _2F_1\left(\frac{1}{2},\frac{1}{2};\frac{5}{2};v^2\right)-45 \pi }{15 \pi  (3 v+1)^2}.
\end{equation}
We can upper bound $v\leq |v|\leq \gamma$ for any $v\in (-1/3, \gamma]$ as long as $\gamma\geq 1/3$.  Hence, by the series representation:
\begin{equation}\label{eq:hyper_geo}
    \, _2F_1(a, b;c;z^2)=\sum_{j=0}^\infty \frac{(a)_j (b)_j}{(c)_j} z^{2j},
\end{equation}
we can upper bound $h(v, v)\leq h(v, \gamma)$ for any $v\in(-1/3, \gamma]$.

%$\,_2 F_1$ is monotonically increasing as a function of $|z|$, so we can upper bound $h(v, v) \leq h(v,\gamma)$ for any $v\in [-1/3, \gamma]$ as long as $\gamma\geq 1/3$
%For fixed $a, b, c$ as $z$ increases the value of each term and hence increasing the overall function value. 
%\iffalse
%\begin{equation}\label{eq:bound_on_der}
 %   \leq \frac{8 (3 \gamma+1) \gamma^2 \, _2F_1\left(\frac{3}{2},\frac{3}{2};\frac{7}{2};\gamma^2\right)+40 \, _2F_1\left(\frac{1}{2},\frac{1}{2};\frac{5}{2};\gamma^2\right)-45 \pi }{15 \pi  (3 v+1)^2}
%\end{equation}
%\fi
%For any $\gamma$ in the stated interval the R.H.S. of \Cref{eq:bound_on_der} is negative for all choices of $v$.  
For any $\gamma$ that satisfies the hypothesis, $h(v,\gamma)$ is negative for all choices of $v$. Hence, $g(v)$ is decreasing and we can write:
\begin{equation*}
\min_{v \in (-1/3,\gamma]} \frac{1 + F(3, v)}{1 + 3v} \geq  \frac{1 + F(3, \gamma)}{1 + 3\gamma}.
\end{equation*}
Further, we may assume the R.H.S. is $\geq 1/2$ by hypothesis.

%We can also numerically check that $g(0.915)>1/2$ so the lemma follows.
\end{proof}

\begin{definition}[The value $\gamma$]
For any fixed $\gamma$, one may numerically evaluate the numerator of $h(v,\gamma)$ to check that it satisfies the hypothesis of \Cref{lem:1/2-approx}.  For the remainder of this paper we set $\gamma := 0.911$, which does so.  For our analysis we are free to pick any value which satisfies both \Cref{lem:1/2-approx} (\Cref{fact:1}) and \Cref{fact:2}.  
\end{definition}

\subsection{Performance of Rounding Algorithm}

\begin{definition}[Edge types for \Cref{alg:main}]
There are three types of edges in reference to \Cref{alg:main}:
\begin{enumerate}%[label=(\roman*)]
    \item The large edges: $L = \{ij \in E \mid v_{ij} > \gamma\}$
    \item Edges adjacent to $L$ (but not in $L$): $\delta(L)$ (see \Cref{sec:preliminaries})
    \item The remaining edges, those in $E(B) := \{ij \in E \mid i \in B\text{ and }j \in B\}$.
\end{enumerate}
\end{definition}

We have a $1/2$-approximation, satisfying \Cref{prop:approx}, for any edge in $L$:
\begin{equation} \label{eq:1/2-approx-L}
  \tr\left[\frac{\Id - X\otimes X - Y \otimes Y - Z \otimes Z}{4}\rho\right] = \frac{1}{2} \geq \frac{1}{2}\mu_{ij},
\end{equation}
where $\mu_{ij} \in [0,1]$ by \Cref{lem:QMC-values-ranges}, and $\rho = (\Id - Z \otimes Z)/4$. We also have a $1/2$-approximation for edges in $E(B)$ by \Cref{fact:1} (\Cref{lem:1/2-approx}).

%\begin{fact}[by \Cref{lem:1/2-approx}]
%\label{fact:1}
 %   If $H$ is a QMC Hamiltonian, and if edge $ij$ has $v_{ij} \leq \gamma$ then product state rounding has approximation factor $\geq 1/2$ on $ij$.
%\end{fact}

This leaves edges in $\delta(L)$.  We show that such edges cannot have a large value.  Suppose $ij \in L$ and $ik \in E$.  The Lieb-Mattis bound used by Anshu, Gosset, and Morenz~\cite{A20} and Parekh and Thompson~\cite{PT21} yields that $v_{ij} + v_{ik} \leq 4/3$ (see Equation~(32) in~\cite{PT21}).  Although this bound can be tight (e.g., when $v_{ij} = v_{ik} = 2/3$), it is not strong enough for our purposes because it can only be tight for limited values of $v_{ij}$.  We take a bounding approach that takes the region of feasible $v_{ij}$ and $v_{ik}$ into account.

%\begin{fact}[by \Cref{lem:large-edge-bound}]
%\label{fact:2}
 %   If $H$ is a QMC Hamiltonian, and if an edge $ij$ has $v_{ij} > \gamma$ then any adjacent edge $ik$ has $v_{ik} < 1/3$.
%\end{fact}

\begin{remark}[$L$ is a matching]
\label{rem:matching}
A consequence of \Cref{fact:2} is that the edges in $L$ must be a matching, since we take $\gamma = 0.911 \geq 1/3$. 
\end{remark}

First we assume \Cref{fact:2} and wrap up the proof of \Cref{thm:main}.
\begin{proof}[Proof of \Cref{thm:main}]
By \Cref{eq:1/2-approx-L} and \Cref{fact:1}, it remains to demonstrate a $1/2$-approximation for edges in $\delta(L)$.  Choose $ij \in \delta(L)$.  By \Cref{fact:2}, $v_{ij} < 1/3$ since $ij$ is adjacent to some edge in $L$.  By \Cref{def:QMC-values}, this means $1/3 > v_{ij} = (4\mu_{ij}-1)/3$, which implies $\mu_{ij} < 1/2$. 

Now let $i$ be the endpoint of $ij$ adjacent to $ik \in L$.  In the solution $\rho$ returned by \Cref{alg:main}, the only nontrivial Pauli term involving $i$ or $k$ is a single term proportional to $Z_iZ_k$, since $L$ is a matching.  There may be nontrivial Pauli terms involving $j$ that are tensored with the $Z_iZ_k$ term; however, this cannot result in a degree-$2$ term on $ij$.  Thus we have
\begin{equation*}
    \tr\left[\frac{\Id - X_iX_j - Y_iY_j - Z_iZ_j}{4}\rho\right] = \frac{1}{4} > \frac{1}{2}\mu_{ij}.
\end{equation*}
\end{proof}

\subsection{Monogamy of Entanglement on a Triangle}
\label{sec:nonlinear-triangle-inequality}
We now turn to the proof of \Cref{fact:2}.  We will analyze an arbitrary triangle of edges and assume without loss of generality that we are considering the three edges on qubits 1, 2, and 3.  In analyzing an arbitrary $\las_2$ solution for QMC, we will find it convenient to work with a different operator than the usual Paulis.  We define the unitary SWAP operator on the edge $ij$:
\begin{equation*}
    S_{ij}=\frac{\mathbb{I} + X_i X_j + Y_i Y_j + Z_i Z_j}{2}.
\end{equation*}

Consider a Gram matrix using the vectors $\ket{\mathbb{I}}$, $\ket{S_{12}}$, $\ket{S_{13}}$, $\ket{S_{23}}$, which are derived from a feasible $\las_2$ solution via \Cref{def:polynomial-vector}.  In other words, we define the matrix $M_{\C} \in \C^{4 \times 4}$ as $M_{\C}(P,Q) := \braket{P|Q}$ with $P,Q \in \{\Id,S_{12},S_{13},S_{23}\}$.  We will find that considering only the real part of $M_{\C}$ offers additional simplicity.  In particular we define $M := (M_{\C}+M_{\C}^T)/2 \in \R^{4 \times 4}$.  We have $M_{\C} \succeq 0$ by construction, which implies that $M \succeq 0$ as well.  We get
\begin{equation*}
    M(P,Q) = \frac{\braket{P|Q} + \braket{Q|P}}{2} = \frac{\braket{PQ}+\braket{QP}}{2} = \frac{\braket{PQ + QP}}{2},
\end{equation*}
by \Cref{lem:lasserre-polynomial-product} and \Cref{def:polynomial-value}.
Letting $s_{ij} = \braket{S_{ij}} \in [-1,1]$ (\Cref{def:QMC-values}) and using the above equation, the entries of $M$ are derived as:
\begin{equation*}
M(S_{ij},S_{ij}) = \frac{\braket{2S_{ij}^2}}{2} = \braket{\Id} = 1,
\end{equation*}
and when $i,j,k$ are distinct,
\begin{alignat*}{2}
 M(S_{ij},S_{ik})
&= \mathrlap{\frac{1}{2} \left\langle \frac{\Id + X_iX_j + Y_iY_j + Z_iZ_j}{2}  \frac{\Id + X_iX_k + Y_iY_k + Z_iZ_k}{2} + \right.} \\
& \mathrlap{\phantom{=\left\langle\frac{1}{2}\right.} \left. \frac{\Id + X_iX_k + Y_iY_k + Z_iZ_k}{2} \frac{\Id + X_iX_j + Y_iY_j + Z_iZ_j}{2} \right\rangle} \\
&= \mathrlap{\frac{1}{2}\left\langle \frac{\Id + X_iX_j + Y_iY_j +Z_iZ_j + X_iX_k + Y_iY_k + Z_iZ_k + X_jX_k + Y_jY_k + Z_jZ_k}{2} \right\rangle}\\
&= \frac{\braket{S_{ij}} + \braket{S_{ik}} + \braket{S_{jk}} - \braket{\Id}}{2} \qquad && \text{[ by \Cref{def:polynomial-value} ]}\\
&= \frac{s_{ij} + s_{ik} + s_{jk} - 1}{2}
  && \text{[ by \Cref{def:QMC-values} ]}.
\end{alignat*}

\begin{remark}[$\las_2$ and SWAP operators]
In light of the connection between $\las_2$ and psuedo-density operators, we may equivalently view the above as computing expectation values of a pseudo-density operator with respect to SWAP operators rather than Pauli terms. 

While $S_{ij}S_{ik}$ has 3-local terms, these cancel in $S_{ij}S_{ik} + S_{ik}S_{ij}$.  Our approach can be used to prove the same monogamy of entanglement bounds using $\las_2$ as~\cite{PT21}; however, using SWAP operators it suffices to only consider a single vector per edge rather than three per edge as~\cite{PT21}.  Moreover, SWAP operators induce symmetries in $M$ that are not explicit in traditional Pauli-based $\las_2$ analysis.  This is a key feature that enables our results.
\end{remark}

For further notational convenience, we let $p:=s_{12}$, $q:=s_{13}$, $r:=s_{23}$, and $s:=p+q+r$. We now have:
\begin{equation*}
M = \begin{bmatrix}
1 & p & q & r\\
p & 1 & \frac{s-1}{2} & \frac{s-1}{2}\\
q & \frac{s-1}{2} & 1 & \frac{s-1}{2}\\
r & \frac{s-1}{2} & \frac{s-1}{2} & 1
\end{bmatrix} \succeq 0.
\end{equation*}
We will derive our desired bounds by analyzing $M$.

\paragraph{Nonlinear triangle inequality.}  First we show that the condition $M \succeq 0$ can be captured by a quadratic polynomial inequality in $p$, $q$, and $r$.
\begin{lemma}For $M$ as defined above, $M \succeq 0$ is equivalent to
\begin{gather}
\label{eq:linear-triangle} 0 \leq p+q+r \leq 3\\
\label{eq:nonlinear-triangle} p^2 + q^2 + r^2 + 2(p+q+r) - 2(pq+pr+qr) \leq 3
\end{gather}
\end{lemma}
\begin{proof}
First, we write $M$ as a block matrix
\begin{equation*}
M = \begin{bmatrix}
1 & b^T\\
b & C 
\end{bmatrix}\text{ with }
C = \begin{bmatrix}
1 & \frac{s-1}{2} & \frac{s-1}{2}\\
\frac{s-1}{2} & 1 & \frac{s-1}{2}\\
\frac{s-1}{2} & \frac{s-1}{2} & 1
\end{bmatrix}\text{ and }
b^T = (p,q,r).
\end{equation*}
We will employ a Schur complement to reduce the positivity of $M$ to simpler conditions.  By Theorem 4.3 in~\cite{G10} we have
\begin{equation*}
    M \succeq 0 \quad\Longleftrightarrow\quad \text{(i) } C \succeq 0,\ \text{(ii) } (\Id - CC^+)b = \vec{0},\ \text{(iii) } 1 - b^T C^+ b \geq 0, 
\end{equation*}
where $C^+$ is the pseudoinverse of $C$. To employ (i)-(iii), we obtain an eigendecomposition of $C$.  One may verify that we may pick the following eigenvectors for $C$, with associated eigenvalues:
\begin{alignat*}{2}
    u_1 &= (1,1,1)^T, &&\lambda_1 = s,\\
    u_2 &= (-1,1,0)^T,\quad&&\lambda_2 = \frac{3-s}{2},\\
    u_3 &= (1,1,-2)^T, &&\lambda_3 = \frac{3-s}{2}.
\end{alignat*}
From the above, we derive projectors onto the subspaces corresponding to the two distinct eigenvalues; this allows us to express $C = \lambda_1 \Pi_1 + \lambda_2 \Pi_2$, where
\begin{equation*}
\Pi_1 = \frac{1}{3}\begin{bmatrix}
1 & 1 & 1\\
1 & 1 & 1\\
1 & 1 & 1
\end{bmatrix}\text{ and }
\Pi_2 = \frac{1}{3}\begin{bmatrix}
2 & -1 & -1\\
-1 & 2 & -1\\
-1 & -1 & 2
\end{bmatrix}.
\end{equation*}
The matrix $\Pi_1$ is a rank-$1$ projector, and $\Pi_2$ must be the desired rank-$2$ projector since $\Pi_2 = \Id - \Pi_1$.  

Property (i) is equivalent to $\lambda_1,\lambda_2 \geq 0$, which in turn is equivalent to $s \in [0,3]$ (i.e., \Cref{eq:linear-triangle}).  If $s \in (0,3)$ then $C$ is full rank, and (ii) becomes vacuous.  In this case, scaling (iii) by a factor of 3 yields 
\begin{equation}\label{eq:triangle-polynomial}
    0 \leq 3 - \frac{3}{\lambda_1}b^T\Pi_1 b - \frac{3}{\lambda_2}b^T \Pi_2 b = 3 - \frac{1}{s}(p+q+r)^2 - \frac{4}{3-s}(p^2 + q^2 + r^2 - pq - pr - qr).
\end{equation}
Algebraic manipulation demonstrates that the above is indeed equivalent to \Cref{eq:nonlinear-triangle}.

The remaining cases are $s \in \{0,3\}$.  When $s=3$, $C^+ = \frac{1}{3}\Pi_1$ and $CC^+ = \Pi_1$.  Consequently, (ii) is equivalent to $b = (1,1,1)$ since $s = 3$.  In this case adding $0=(3-p-q-r)^2$ to the left-hand side of \Cref{eq:nonlinear-triangle} shows that the latter is equivalent to $p^2 + q^2 + r^2 \leq 3$.  Since $s=3$, we must have $p=q=r=1$.  Thus both $M\succeq 0$ and \Cref{eq:linear-triangle,eq:nonlinear-triangle} are equivalent to $p=q=r=1$.

Finally, when $s=0$, $\Id-CC^+ = \Pi_1$ and consequently (ii) is equivalent to $s=0$.  In this case, for (iii), the second term of \Cref{eq:triangle-polynomial} is absent, yielding $\frac{4}{3}(p^2 + q^2 + r^2 - pq - pr - qr) \leq 3$.  Subtracting $0 = \frac{1}{3}(p+q+r)^2$ from the left-hand side shows that this is equivalent to \Cref{eq:nonlinear-triangle}.

\end{proof}

\Cref{eq:nonlinear-triangle} captures the the precise relationship between the values $p$ and $q$ on two adjacent edges we require to prove \Cref{fact:2}.  However, below we opt for a more direct analysis of $M$, allowing us to employ SDP duality to extract the desired relationship.

\begin{lemma} \label{lem:large-edge-bound}
For $M$ as defined above, if $q \leq 0$, then $p \geq -\sqrt{3}/2$.
\end{lemma}
\begin{proof}We seek to show that the minimum feasible value of $p$, while constraining $q \leq 0$, is at least $-\sqrt{3}/2$.  This is captured by an SDP:
\begin{alignat*}{2}
\beta := \min\ p&\\
s.t. \qquad
    pM_1 + qM_2 + rM_3 + B &\succeq 0\\
    q &\leq 0, 
\end{alignat*}
where (zeroes omitted for readability):
\begin{equation*}
M_1 = \begin{bmatrix}
 & 1 &  & \\
1 &  & \frac{1}{2} & \frac{1}{2}\\
 & \frac{1}{2} &  & \frac{1}{2}\\
 & \frac{1}{2} & \frac{1}{2} & 
\end{bmatrix},\ 
M_2 = \begin{bmatrix}
 &  & 1 & \\
 &  & \frac{1}{2} & \frac{1}{2}\\
1 & \frac{1}{2} &  & \frac{1}{2}\\
 & \frac{1}{2} & \frac{1}{2} & 
\end{bmatrix},\ 
M_3 = \begin{bmatrix}
 &  &  & 1\\
 &  & \frac{1}{2} & \frac{1}{2}\\
 & \frac{1}{2} &  & \frac{1}{2}\\
1 & \frac{1}{2} & \frac{1}{2} & 
\end{bmatrix}, \text{ and }
B = \begin{bmatrix}
1 &  &  & \\
 & 1 & -\frac{1}{2} & -\frac{1}{2}\\
 & -\frac{1}{2} & 1 & -\frac{1}{2}\\
 & -\frac{1}{2} & -\frac{1}{2} & 1 
\end{bmatrix}.
\end{equation*}
The dual SDP is:
\begin{align*}
\beta' := \max\ \tr[-B X]&\\
s.t. \qquad 
 \tr[M_1 X] &= 1\\
 \tr[M_2 X] &\geq 0\\
 \tr[M_3 X] &= 0\\
 X &\succeq 0.
\end{align*}
Due to the primal inequality constraint, $q \leq 0$, this pair of SDPs is not in standard form~\cite{BV96}.  However, dropping the inequality in the primal and replacing the matrices with $M_1':=0 \oplus M_1$, $M_2' := -1 \oplus M_2$, $M_3' := 0 \oplus M_3$, and $B' := 0 \oplus B$ offers an equivalent standard pair of SDPs.  

We construct an explicit dual solution certifying the desired bound. Let 
\begin{equation*}
x_1 := [\sqrt{3},2,0,1]^T,\ x_2 := [-\sqrt{3},-1,1,0]^T,\ \alpha_1 = \frac{\sqrt{3}}{4}-\frac{1}{3},\text{ and }\alpha_2 = \frac{1}{12+6\sqrt{3}}.
\end{equation*}
We set $X := \alpha_1 \cdot x_1x_1^T + \alpha_2 \cdot x_2x_2^T$.  One may verify that $\tr[B X] = \sqrt{3}/2$, $\tr[M_1 X]=1$, $\tr[M_2 X] = 1/2 \geq 0$, and $\tr[M_3 X] = 0$, hence $X$ is a dual feasible solution.  We have $-\sqrt{3}/2 = \beta' \leq \beta \leq p$, by weak duality.

\end{proof}

\begin{remark}  
One may wonder why orthogonal $x_1$ and $x_2$ were not chosen above. The above values were chosen to simplify verification; orthogonalization results in significantly more complicated expressions.  Likewise, a rank-1 solution $X$ exists but is unwieldy to express exactly.  In case the reader is curious, an optimal primal solution, satisfying complementary slackness with $X$, is $p=-\sqrt{3}/2$, $q=0$, and $r=\sqrt{3}/2$.  
\end{remark}

We are now in a position to prove \Cref{fact:2}.
\begin{proof}[Proof of \Cref{fact:2}]
The contrapositive of \Cref{lem:large-edge-bound} establishes that if $p_{ij} < -\sqrt{3}/2$ then $p_{ik} > 0$. 
From the relationships in \Cref{def:QMC-values}, we see that $p_{ij} = (1-3v_{ij})/2 < -\sqrt{3}/2$ is equivalent to $v_{ij} > (1+\sqrt{3})/3$. Likewise, we have $p_{ik} = (1-3v_{ik})/2 > 0$ is equivalent to $v_{ik} < 1/3$.  Thus if $v_{ij} > \gamma = 0.911 > (1+\sqrt{3})/3$, then $v_{ik} < 1/3$.  
\end{proof}

\section*{Acknowledgements}
We thank John Wright for suggesting that product-state approximations improving Gharibian and Parekh~\cite{G19} might be possible.  Sandia National Laboratories is a multimission laboratory managed and operated by National Technology and Engineering Solutions of Sandia, LLC., a wholly owned subsidiary of Honeywell International, Inc., for the U.S. Department of Energy’s National Nuclear Security Administration under contract DE-NA-0003525. This work was supported by the U.S. Department of Energy, Office of Science, Office of Advanced Scientific Computing Research, Accelerated Research in Quantum Computing. O.P.~was also supported by U.S. Department of Energy, Office of Science, National Quantum Information Science Research Centers.

\bibliography{mybib}
\bibliographystyle{alpha}

\appendix 
\section{The Quantum Lasserre Hierarchy}\label{sec:lasserre}

 We find it clean and convenient to express the Laserre hierarchy from the perspective of a vector program.  The connection to the psuedo-density formulation above will be made explicit at the end of the section.

\onote{Give informal explanation of vector version of Lasserre here}

\begin{mdframed}
\begin{definition}[Laserre vector solution] \label{def:vector-solution}
A feasible solution to $\las_k$, denoting level $k$ of the Lasserre hierarchy, is expressed as a collection of complex unit vectors:
\begin{equation*}
\ket{\tau} \in \C^d,\text{ for each }\tau \in \mathcal{P}_n(k).
\end{equation*}
We employ bra-ket notation for vectors to streamline notation and suggest a connection to moments of quantum states.  We extend the above definition to Pauli polynomials rather than just Pauli monomials by linearity.  For any Pauli polynomial of degree at most $k$,
\begin{equation*}
    P = \sum_{\tau \in \mathcal{P}_n(k)} p(\tau)\cdot \tau,\text{ with }p(\tau) \in \R\ \forall \tau,
\end{equation*}
we define
\begin{equation} \label{def:polynomial-vector}
    \ket{P} := \sum_{\tau \in \mathcal{P}_n(k)} p(\tau)\ket{\tau}.
\end{equation}
Thus we view a $\las_k$ solution as providing vectors for Pauli polynomials of degree at most $k$. 
\end{definition}
\end{mdframed}

\begin{mdframed}
\begin{definition}[Value of a Lasserre vector solution]
\label{def:vector-solution-value}
From the vectors of a $\las_k$ solution, we may derive relaxed values for Pauli polynomials of degree at most $2k$.  Express $\phi \in \mathcal{P}_n(2k)$ as a product of non-trivial Paulis acting on single qubits, $\phi = \phi_{i_1}\phi_{i_2}\cdots\phi_{i_l}$ with $i_1 < \cdots < i_l$. Let $h = \min\{k,l\}$; then the \emph{left part} and \emph{right part} of $\phi$ are 
\begin{align*}
\phi_{L} &:= \phi_{i_1}\ldots\phi_{i_h}, \text{ and }\\
\phi_{R} &:= \phi_{i_{h+1}}\ldots\phi_{i_l},
\end{align*}
respectively.  We may now define the \emph{value} of $\phi$ as 
\begin{equation} \label{def:term-value}
    \braket{\phi} := \braket{\phi_L | \phi_R}.
\end{equation}
We extend the notion of value, by linearity, to Pauli polynomials of degree at most $2k$.  In this case we allow complex coefficients for convenience. For
\begin{equation*}
    Q = \sum_{\phi \in \mathcal{P}_n(2k)} q(\phi)\cdot \phi,\text{ with }q(\phi) \in \C\ \forall \phi,
\end{equation*}
we define
\begin{equation} \label{def:polynomial-value}
    \braket{Q} := \sum_{\phi \in \mathcal{P}_n(2k)} q(\phi)\braket{\phi}.
\end{equation}
If $Q \in \herm_n$ then we may assume $q(\phi) \in \R\ \forall \phi$, hence $\braket{Q} \in \R$, while if $Q \notin \herm_n$ then $\braket{Q} \in \C$.
\end{definition}
\end{mdframed}

\begin{mdframed}
\begin{definition}[Quantum Lasserre hierarchy] \label{def:lasserre}
We are given as input $H = \sum_{\phi \in \mathcal{P}_n(2k)} c(\phi) \phi$, with $H \in \herm_n$.  Level $k$ of the Quantum Lasserre hierarchy, denoted $\las_k$, is defined by the following vector program:
\begin{alignat}{2}
\nonumber 
    \nu_k(H) := \max \sum_{\phi \in \mathcal{P}_n(2k)} & c(\phi) \braket{\phi}\\
    s.t. \qquad
    %\label{cons:lasserre-identity}
    %\braket{\Id | \Id} &= 1 \\
    \label{cons:lasserre-unit}
    \braket{\tau | \tau} &= 1
        \quad && \forall \tau \in \mathcal{P}_n(k) \\
    \label{cons:lasserre}
    \braket{\tau | \sigma} &= \braket{\tau\sigma} 
        \quad && \forall \tau,\sigma \in \mathcal{P}_n(k) \\
\nonumber
    \ket{\tau} &\in \C^d
        \quad && \forall \tau \in \mathcal{P}_n(k),
\end{alignat}
for any integer $d \geq |\mathcal{P}_n(k)|$.
\end{definition}
\end{mdframed}

\begin{remark}[Quantum Lasserre hierarchy] \label{rem:lasserre}
Some remarks regarding \Cref{def:lasserre} are in order:
\begin{enumerate}%[label=(\roman*)]
    \item The intuition driving $\las_k$ is that $\ket{\tau}$ is meant to approximate $\tau\ket{\psi}$ for a true quantum state $\ket{\psi}$.  This connection is made more precise in \Cref{lem:relaxation-bound}, that establishes $\las_k$ is indeed a relaxation.

    \item \label{item:lasserre-(k-1)} A feasible solution for $\las_{k-1}$ may be extracted from $\las_{k}$ by considering only the vectors $\ket{\tau}$ with $\tau \in \mathcal{P}_n(k-1)$.

\begin{comment}
    \item For Constraint~\eqref{cons:lasserre}, if $\tau\sigma \in \herm_n$ then $\braket{\tau\sigma} \in \R$ as pointed out in \Cref{def:vector-solution}. However, when $\tau\sigma \notin \herm_n$ then $\tau\sigma = -\sigma\tau$, and we have $\braket{\tau | \sigma} = \braket{\tau\sigma} = -\braket{\sigma\tau} = -\braket{\sigma | \tau}$,
    implying
    \begin{equation} \label{cons:lassere-orthogonal}
    \braket{\tau | \sigma} = 0\text{ if }\tau\sigma \notin \herm_n.
    \end{equation}
    A complex version of $\las_k$ may also be defined in which the $\ket{\tau}$ are complex unit vectors.  We find the real version easier to analyze. See~\cite{PT21} for more details.
\end{comment}

    \item The vector program $\las_k$ may be solved as a complex SDP by replacing the vectors $\ket{\tau}$ with the matrix $M \in \C^{\mathcal{P}_n(k) \times \mathcal{P}_n(k)}$ where $M(\tau,\sigma) := \braket{\tau | \sigma}$.  The entries of $M \succeq 0$ are variables of the SDP.  The constraints \eqref{cons:lasserre} and \eqref{cons:lasserre-unit} become linear equalities, with complex coefficients, in the entries of $M$.  Finally, the objective is linear in these entries, hence $\las_k$ is equivalent to a complex SDP.  Complex SDPs may be solved as real SDPs by doubling the dimensions of the matrices.  This gives us a real SDP over matrices of size $O(n^k) \times O(n^k)$. See~\cite{PT21} for more details.
    
    \item Although in \Cref{def:lasserre} we do not specify a value of $d$, we may always construct a vector solution using a Cholesky decomposition of the $|\mathcal{P}_n(k)| \times |\mathcal{P}_n(k)|$ SDP solution matrix $M$ described above.  Thus we may convert any feasible solution into one with $d = |\mathcal{P}_n(k)|$, such that the values are preserved.
    
    \item \label{rem:real-lasserre}
    One may also define a real version of $\las_k$ in which the $\ket{\tau}$ are real, and Constraint~\eqref{cons:lasserre} is modified so $\braket{\tau\sigma} = 0$, if $\tau\sigma \notin \herm_n$.  The objective function remains the same since $\mathcal{P}_n(2k) \subset \herm_n$.  A $2d$-dimensional real solution $\{\ket{\tau'}\}$ may be obtained from a $d$-dimensional complex one $\{\ket{\tau}\}$ by setting $\ket{\tau'} := \Re(\ket{\tau}) \oplus \Im(\ket{\tau})$, where $\Re$ and $\Im$ refer to real and imaginary parts.  From the moment matrix perspective, this is equivalent to setting $M' := (M+M^T)/2$.  Real solutions obtained in this way preserve the objective value.
\end{enumerate}
\end{remark}

\subsection{Properties of the Quantum Lasserre Hierarchy}

We warm up with some fundamental properties of $\las_k$ solutions.  First we demonstrate that the objective $\nu_k(H)$ given in \Cref{def:relaxed-objective} is indeed a relaxation of $\lambda_{\max}(H)$.
\begin{lemma} \label{lem:relaxation-bound}
Let $H \in \herm_n$.  Then $\las_k$ gives a relaxation of $\lambda_{\max}(H)$, namely $\nu_k(H) \geq \lambda_{\max}(H)$.  
\end{lemma}
\begin{proof}
Pick any unit $\ket{\psi} \in (\C^2)^{\otimes n}$ such that $\braket{\psi|H|\psi} = \lambda_{\max}(H)$.  We show that setting $\ket{\tau} := \tau \ket{\psi}$ for each $\tau \in \mathcal{P}_n(k)$ yields a feasible solution for $\las_k$.
\begin{comment}
However, since we opt for the real version of $\las_k$, we must consider the real and imaginary parts of $\ket{\psi}$ separately. Take $\ket{\psi_r}, \ket{\psi_i} \in (\R^2)^{\otimes n}$ such that $\ket{\psi} = \ket{\psi_r} + i\ket{\psi_i}$.

For each Pauli term $\tau \in \mathcal{P}_n(k)$, set $\ket{\tau} := \tau\ket{\psi_r} \oplus \tau\ket{\psi_i}$. We will show that $\{\ket{\tau}\}_{\tau \in \mathcal{P}_n(k)}$ is a feasible $\las_k$ solution.  First observe that for any $\tau$, $\braket{\tau | \tau} = \braket{\psi_i | \psi_i} + \braket{\psi_r | \psi_r} = 1$.  
\end{comment}

First we observe that for any Pauli term $\phi \in \mathcal{P}_n(2k)$, applying \Cref{def:term-value} with respect to our $\{\ket{\tau}\}$ implies
\begin{equation} \label{eq:psi-bar}
\braket{\psi|\phi|\psi}
= \braket{\psi|\phi_L \phi_R|\psi} 
= \braket{\phi_L | \phi_R}
= \braket{\phi}.
\end{equation}

Constraint~\eqref{cons:lasserre-unit} is satisfied since $\braket{\tau |\tau} = \braket{\psi|\tau^2 |\psi} = \braket{\psi|\psi} = 1$.  To check Constraint~\eqref{cons:lasserre}, suppose $\tau,\sigma \in \mathcal{P}_n(k)$ with $\tau\sigma = c \phi$, where $c \in \{\pm 1, \pm i\}$ and $\phi \in \mathcal{P}_n(2k)$.  We get
\begin{alignat*}{2}
    \braket{\tau | \sigma} = \braket{\psi |\tau \sigma |\psi} = c\braket{\psi|\phi |\psi} &= c\braket{\phi} \qquad && \text{[ by \Cref{eq:psi-bar} ]}\\
    &= \braket{\tau\sigma} \qquad && \text{[ by \Cref{def:polynomial-value} ]}.
\end{alignat*}

Finally, for the objective value express $H$ as $\sum_{\phi \in \mathcal{P}_n(2k)} c(\phi)\cdot \phi$.  Another application of \Cref{eq:psi-bar} yields
\begin{equation*}
\braket{\psi|H|\psi} = 
\sum_{\phi \in \mathcal{P}_n(2k)} c(\phi) \braket{\psi|\phi|\psi} =
\sum_{\phi \in \mathcal{P}_n(2k)} c(\phi) \braket{\phi}.
\end{equation*}
Since we constructed a feasible solution with objective value $\braket{\psi|H|\psi}$, we have $\nu_k(H) \geq \braket{\psi|H|\psi} = \lambda_{\max}(H)$.
\end{proof}

Next we show that Constraint~\eqref{cons:lasserre} extends more generally to Pauli polynomials rather than just Pauli terms.
\begin{lemma} \label{lem:lasserre-polynomial-product}
For a feasible $\las_k$ solution $\{\ket{\tau}\}_{\tau \in \mathcal{P}_n(k)}$ and Pauli polynomials $P,Q$ of degree at most $k$, we have $\braket{P|Q} = \braket{PQ}$.
\end{lemma}
\begin{proof}
Let $P = \sum_{\tau \in \mathcal{P}_n(k)} p(\tau)\cdot \tau$, and $Q = \sum_{\sigma \in \mathcal{P}_n(k)} q(\sigma)\cdot \sigma$. Then
\begin{alignat*}{2}
\braket{P|Q} &= \left(\sum_{\tau \in \mathcal{P}_n(k)} p(\tau)\bra{\tau} \right) \left(\sum_{\sigma \in \mathcal{P}_n(k)} q(\sigma)\ket{\sigma} \right)
\qquad && \text{[ by \Cref{def:polynomial-vector} ]}\\
&= \sum_{\tau,\sigma \in \mathcal{P}_n(k)} p(\tau)q(\sigma)\braket{\tau\sigma} && \text{[ by Constraint~\eqref{cons:lasserre} ]}\\
&= \braket{PQ}. && \text{[ by \Cref{def:polynomial-value} ]}
\end{alignat*}
\end{proof}

Finally, we show that a $\las_k$ solution may be ``rotated'' by tensor products of single-qubit unitaries to produce another feasible $\las_k$ solution.  The rotation corresponds to a change of the Pauli basis for each qubit.

\begin{lemma}\label{lem:rotation_L2} Let $U \in \mathcal{U}_n$ be a tensor product of single-qubit unitaries.  For a feasible $\las_k$ solution $\{\ket{\tau}\}_{\tau \in \mathcal{P}_n(k)}$, let $\ket{\tau'}:=\ket{U^\dagger \tau U}\ \forall \tau$.  Then $\{\ket{\tau'}\}_{\tau \in \mathcal{P}_n(k)}$ is a feasible solution for $\las_k$.
\end{lemma}
\begin{proof} Each $\tau'$ is Hermitian and may represented by a Pauli polynomial, so each $\ket{\tau'}$ is well defined by \Cref{def:polynomial-vector}.  Moreover, since $U$ is a tensor product of single-qubit unitaries, $\tau'$ has the same degree as $\tau$, so each $\tau'$ has degree at most $k$.  By \Cref{lem:lasserre-polynomial-product}, $\braket{\tau'|\sigma'} = \braket{\tau'\sigma'}$, and $\braket{\tau'|\tau'} = \braket{U^\dagger \tau UU^\dagger \tau U} = \braket{\Id} = \braket{\Id | \Id} = 1$; hence, the constraints of $\las_k$ are satisfied.
\end{proof}

We could replace $U$ above by an arbitrary unitary in $\mathcal{U}_n$.  In this case, the resulting solution would be feasible for the Lasserre hierarchy, but not necessarily for level $k$, since the degree of $\tau'$ could be larger than $k$. 
 
\paragraph{Pseudo-densities.}  We close our introduction to $\las_k$ by sketching the connection to psuedo-densities (\Cref{def:psuedo-density}).  We may extract a psuedo-density operator $\psuedo{\rho}$ from a $\las_k$ solution as
\begin{equation*}
    \psuedo{\rho} := \frac{1}{2^n}\sum_{\phi \in \mathcal{P}_n(2k)} \braket{\phi} \phi,
\end{equation*}
so that $\braket{\phi} = \tr[\psuedo{\rho}\,\phi], \forall \phi$.  The operator $\psuedo{\rho}$ is indeed a $k$-positive pseudo-density because $\tr[\psuedo{\rho}] = \braket{\Id} = \braket{\Id | \Id} =  1$, and $\tr[\psuedo{\rho}P^2] = \braket{P^2} = \braket{P | P} \geq 0$, by \Cref{lem:lasserre-polynomial-product}, for any Pauli polynomial $P$ of degree at most $k$.

Conversely, given a $k$-positive psuedo-density operator $\psuedo{\rho}$, define a moment matrix $M \in \C^{\mathcal{P}_n(k) \times \mathcal{P}_n(k)}$ with $M(\tau,\sigma) := \tr[\psuedo{\rho}\,\tau\sigma]$, for all $\tau,\sigma \in \mathcal{P}_n(k)$.  For any Pauli polynomial $P$ of degree at most $k$, define a vector $p \in \R^{\mathcal{P}_n(k)}$ with $p(\tau) := \tr[P\,\tau]$.  Then one may verify $p^T M p = \tr[\psuedo{\rho}P^2]$ (e.g., see below Definition 3 in~\cite{PT21}).  Since $\tr[\psuedo{\rho}P^2] \geq 0$, we get $M \succeq 0$.  A solution for $\las_k$ may then be obtained by a Cholesky decomposition of $M$.

Thus we may view $\las_k$ as arising from either of the equivalent optimization problems of \Cref{prob:lasserre-psuedo-density} or \Cref{def:lasserre}.  From the discussion above and \Cref{lem:n-positivity}, we have the following.
\begin{lemma}For a Hamiltonian $H \in \herm_n$, the objective of $\las_n$ is $\nu_n(H) = \lambda_{\max}(H)$.  Moreover, the values of a $\las_n$ solution, for $\tau \in \mathcal{P}_n(n)$ are $\braket{\tau} = \tr[\rho\,\tau]$ for some $\rho \in \D_n$.  
\end{lemma}

\section{Analysis of \Cref{alg:reg_graph} on Unweighted Graphs with Bounded Minimum Degree}

Recall we are given a min-degree-$d$ unweighted graph $G=(V, E)$ and we are interested in proving some approximation $\alpha(d)$ for the QMC instance $H$ corresponding to $G$.  The analysis proceeds by writing $H$ as a sum over star Hamiltonians then by demonstrating that \Cref{alg:reg_graph} has good performance on each star hence it also has good performance overall.  We calculate $\alpha(d)$ by minimizing the approximation factor over the possible values of the edges (\Cref{def:QMC-values}) on the star.  The minimal value leads to a uniform lower bound on uniform lower bound for the approximation factor on each local term and hence yields a lower bound on the approximation factor for the entire Hamiltonian.  The fact that the parameters of the rounding algorithm are coming from $\las_2$ rather than $\las_1$ will imply a crucial constraint which leads to the approximation factor $\alpha(d)$.

The main technical lemma will demonstrate that when optimizing the ``worst-case'' performance occurs when the values on the edges of the star are all the same.  At this point computing $\alpha(d)$ reduces to a simple one parameter optimization which is easily evaluated with calculus.

\onote{$F[x]$ vs $F(x)$ inconsistency between here and elsewhere.  Sorry, I changed to $F(x)$ at some point because I couldn't justify using $F[x]$.}

\begin{proposition}\label{prop:helper}
 Let $F(3, x)$ be as defined in \Cref{def:F[v]}.
\begin{enumerate}
\item\label{item:helper_lemma_item1}  $x-F(3, x) < 2/3 \,\, \forall x\in [-1/3, 1]$.
\item\label{item:helper_lemma_item2}  $F(3, x)$ is convex $\forall x\in [0, 1]$.
\item\label{item:helper_lemma_item3}  $F(3, x)/3x \leq 1/3$ $\forall x\in [-1/3, 0)$.
\item\label{item:helper_lemma_item4}  For $d\geq 2$ an integer, $$
\min_{x\in (-1/3, 1/3+2/(3 d)]} \,\,\frac{1+F(3, x)}{1+3 x}=\frac{1+F(3, 1/3+2/(3 d))}{2+2/d}.
$$
\item\label{item:helper_lemma_item5} Let 
$$
\alpha(d)=\frac{1+F(3, 1/3+2/(3 d))}{2+2/d}.
$$
Then $\alpha(d+1)\geq \alpha(d)$ for $d\geq 2$.  
\end{enumerate}
\end{proposition}
\begin{proof}

First we demonstrate \Cref{item:helper_lemma_item1}.  Note that
$$
|8/(3 \pi)\,_2F_1[1/2, 1/2;5/2;x^2]-1| \leq \max\{8/(3 \pi)\,_2F_1[1/2, 1/2;5/2;1]-1, 1\}=1,
$$
since $\,_2 F_1$ is increasing with $|x|$ by \Cref{eq:hyper_geo}.  Hence,
$$
|x-F(3, x)| \leq |x|,
$$
and \Cref{item:helper_lemma_item1} holds when $x < 2/3$.  For $x\geq 2/3$ we can bound as:
$$
x-F(3, x) \leq 1-F(3, 2/3)\approx 0.4,
$$
hence we have \Cref{item:helper_lemma_item1}.

For \Cref{item:helper_lemma_item2} note that by~\cite{A48}:
$$
\frac{d^2 F}{dx^2}=\left(\frac{56 x \,_2 F_1[3/2, 3/2;7/2;x^2]+24 x^3 \,_2 F_1[5/2, 5/2;9/2;x^2]}{35\pi} \right).
$$
For non-negative $x$ the second derivative is clearly non-negative where defined.  In our interval of interest, the second derivative fails only to be defined at $x=1$.  Continuity of the function itself in the interval $[0, 1]$ and convexity on $[0, 1)$ imply convexity on $[0, 1]$.

For \Cref{item:helper_lemma_item3} we calculate:
$$
\frac{d}{dx}\left( \frac{F(3, x)}{x}\right) =\frac{8 x \,_2 F_1 [3/2, 3/2;7/2;x^2]}{45 \pi}.
$$
This is non-positive in the interval $x\in [-1/3, 0)$ by \Cref{eq:hyper_geo} hence it is maximized at $x=-1/3$.
$$
\frac{F(3, x)}{3x}\bigg\rvert_{x=-1/3} \leq 1/3
$$

To demonstrate \Cref{item:helper_lemma_item4} we calculate:
\begin{equation}
    \frac{d}{dx} \left( \frac{1+F(3, x)}{1+3 x} \right) =\frac{-45 \pi +40 \,_2 F_1[1/2, 1/2;5/2;x^2]+8 x^2(1+3 x) \,_2 F_1[3/2, 3/2;7/2;x^2]}{15\pi (1+3 x^2)}.
\end{equation}
Note that the sign of the derivative only depends on the sign of the numerator.  In the interval $(-1/3, 1/3+2/(3d)] \subseteq (-1/3, 2/3]$ (recall $d\geq 2$) we can upper bound the numerator by taking $x=2/3$ since the numerator is increasing with increasing $|x|$.  The upper bound is:
\begin{equation}
    -45 \pi +40 \,_2F_1[1/2, 1/2;5/2;4/9]+8\cdot{}\frac{4}{9}(1+2) \,_2 F_1[3/2, 3/2;7/2;4/9] \approx -84.
\end{equation}
Hence the derivative is always negative in the interval of interest and the minimum must be the right-most point.  

To prove \Cref{item:helper_lemma_item5} note that the proof of \Cref{item:helper_lemma_item4} demonstrates that $(1+F(3, x))/(1+3 x)$ is non-increasing for $x$ increasing.  Hence it must be non-decreasing for $x$ decreasing.  Taking $x=1/3+2/(3d)$ we can see that $x$ is decreasing as $d$ increases hence $\alpha(d)$ must be non-decreasing for increasing $d$.
\end{proof}

\begin{lemma}\label{lem:reg_main_tech}
Let $d\geq 2$ and define:
\begin{equation}\label{eq:def_of_g}
g(v_1, ..., v_d)=\frac{\sum_{j=1}^d 1+F(3, v_j)}{\sum_{j=1}^d 1+3 v_j}.
\end{equation}
Let:
\begin{align}
\nonumber \alpha(d)=\min g(v_1, ..., v_d)\\
\label{eq:each_v} s.t.\quad -1/3 \leq v_j \leq 1 \,\,\,\forall j,\\
\label{eq:star_bound}\sum_{j=1}^d (1+3 v_j) \leq 2(d+1).
\end{align}
Then,
\begin{equation}
    \alpha(d)=\frac{1+F(3, 1/3+2/(3d))}{2+2/d}
\end{equation}

\end{lemma}
\begin{proof}
First observe that the optimization problem is well-defined.  The only setting of variables which leads to an undefined objective is $v_j=-1/3$ for all $j$.  Near this point it is clear the objective diverges to $+\infty$ since $1+F(3, -1/3)>0$.

The strategy of the proof will be first to show that there exists an optimal solution to $g$ in which all $v_i\geq 0$ and further than in this optimal solution we can assume all $v_i$ have the same value.  After this the value of $\alpha(d)$ will fall from the technical facts (\Cref{prop:helper}) and the constraint \Cref{eq:star_bound}.

Suppose by contradiction that no optimal solution $(v_1^*, ..., v_d^*)$ has $v_i^* \geq 0$ for all $i$.  Of all the optimal solutions let $(v_1^*, ..., v_d^*)$ be a solution with the fewest negative entries.  Assume WLOG that $v_1^* < 0$.  Define:
\begin{align}
C_1=\sum_{j=2}^d 1+F(3, v_j^*)\\
C_2=\sum_{j=2}^d 1+3 v_j^*.
\end{align}

If $C_2=0$, then $v_2^*=v_3^*=...=v_d^*=-1/3$ and the optimizing $\alpha$ satisfies:
\begin{equation}
\alpha=\frac{(d-1) (1+F(3, -1/3))+1+F(3, v_1^*)}{1+F(3, v_1^*)} > 1, 
\end{equation}
since $(1+F(3, -1/3))\approx 0.7$.  However, it is simple to find a better assignment by setting $v_j^*=0$ for all $j$.  Hence we can assume $C_2 >0$.

Denote $v_1^*$ simply as $v$.  By hypothesis,
\begin{equation}
\alpha=\frac{C_1+1+F(3, v)}{C_2+1+3 v} < \frac{C_1+1}{C_2+1}, 
\end{equation}
since otherwise there exists an optimal solution with fewer negative entries.  Since $C_2>0$ and $v<0$, we can rearrange this equation as:
$$
\frac{C_1+1}{C_2+1} < \frac{F(3, v)}{3v}.
$$
So, applying item 3 from \Cref{prop:helper} we derive $\alpha<1/3$.  Using the definition of $g$, this implies that
\begin{equation}
\sum_{j=1}^d 1+F(3, v_j^*) < 1/3 \sum_{j=1}^d 1+3 v_j^*
\end{equation}

Now we can simplify this expression to:
\begin{align}
\frac{2}{3}d < \sum_{j=1}^d v_j^*-F(3, v_j^*)
\end{align}
This implies in particular that there is some $v_j^*$ with $v_j^*-F(3, v_j^*) >2/3$, but this contradicts item $1$ from the lemma.  So we see that there must always be an optimal assignment with $v_i^* \geq 0$ for all $j$. 

We can then infer all $v_i^*$ can be assumed to be the same using convexity of the numerator and linearity of the denominator.  Define $v=\frac{1}{d}\sum_i v_i^*$.  Then,
\begin{align*}
    g(v_1^*, ..., v_d^*)=\frac{d+\sum_j F(3, v_j^*)}{d+3\sum_j v_j^*}=\frac{d+d \sum_j\frac{1}{d} F(3, v_j^*)}{d+3d\sum_j \frac{v_j^*}{d}}\geq \frac{d+d \sum_j F(3, v)}{d+3dv}=\frac{1+F(3, v)}{1+3v},
\end{align*}
where the inequality follows from convexity of $F(3, x)$ (\Cref{item:helper_lemma_item2} of \Cref{prop:helper}).  Note that if $(v_1^*, ..., v_d^*)$ satisfies constraints \Cref{eq:each_v} and \Cref{eq:star_bound} then so does $(v, v, ..., v)$.  Hence, we can assume that $\alpha(d)=g(v, ..., v)$ for some $v$.  For $(v_1, ..., v_d)=(v, ..., v)$ constraint \Cref{eq:each_v} simply translates to $v\in [-1/3, 1]$ while constraint \Cref{eq:star_bound} becomes $v \leq \frac{1}{3}+\frac{2}{3d}$.  So, 
\begin{equation}
    \alpha(d)=\min_{v\in [-1/3, 1/3+2/(3d)]} \,\, \frac{1+F(3, v)}{1+3v}=\frac{1+F(3, 1/3+2/(3d))}{2+2/d},
\end{equation}
where the second equality follows from \Cref{item:helper_lemma_item4} of \Cref{prop:helper}.
\end{proof}

\begin{proof}[Proof of \Cref{thm:approx_reg}]
We will treat the $d\geq 2$ and the $d=1$ cases seperately.  Suppose $d\geq 2$ and let us write the Hamiltonian as $H=1/2 \cdot{}\sum_{i\in V} h_i$ where each $h_i$ is an unweighted star Hamiltonian with $q\geq d$ leaves rooted at $i$.  Let us denote $\rho$ as the expected outcome of \Cref{alg:reg_graph} and for each $i\in V$ let us denote $\{v_1^i, ..., v_q^i\}$ as the values obtained via \Cref{def:QMC-values} for the edges adjacent to $i$. We will demonstrate that $\mathbb{E}[\tr[\rho h_i]] \geq \alpha(d) \sum_j (1+3 v_j^i)/4$, which then implies
\begin{align*}
    \mathbb{E}[\tr[H \rho]]=\frac{1}{2} \sum_{i\in V} \mathbb{E}[\tr[h_i \rho]] \geq \frac{1}{2}\alpha(d) \sum_i \sum_{j\in N(i)}\frac{1+3 v_j^i}{4}=\alpha(d) \sum_{ij\in E} \frac{1+3 v_{ij}}{4} \geq \alpha(d) \lambda_{max}(H).
\end{align*}

By the star bound of~\cite{PT21}, since the SDP relaxation is $\las_2$, for each $i$ $\sum_{j} (1+3 v_j^i) \leq 2(q+1)$.  Additionally we know that $-1/3\leq v_j^i \leq 1$ by \Cref{lem:QMC-values-ranges}.  Note, by \Cref{lem:Briet-et-al},
\begin{equation}\label{eq:reg_last}
    \frac{\mathbb{E}[\tr[\rho h_i]]}{\sum_j (1+3 v_j^i)/4}=\frac{\sum_j 1+F(3, v_j^i)}{\sum_j 1+3 v_j^i}.
\end{equation}
Since the R.H.S. of \Cref{eq:reg_last} is of the same form as \Cref{eq:def_of_g}, and since we have already verified the constraints we may directly apply \Cref{lem:reg_main_tech} and \Cref{prop:helper} to conclude
\begin{equation}
    \frac{\mathbb{E}[\tr[\rho h_i]]}{\sum_j (1+3 v_j^i)/4} \geq \alpha(q)\geq \alpha(d),
\end{equation}
and hence the theorem.  

For the case $d=1$ note that the graph corresponds to a matching.  In this case the SDP will (optimally) choose $v_{ij}=1$ along the edges in the matching.  This means that the Gram vectors from the SDP are anti-aligned across edges in the matching, i.e. $\braket{W_i| W_j}=-1$ if vertices $i$ and $j$ are connected by an edge.  \Cref{alg:reg_graph} will then pick Bloch vectors which are anti aligned and hence achieve objective $1/2$ on each edge.  Hence if there are $m$ edges total the optimal quantum state has energy $m$ and the product state produced by the rounding algorithm has objective $m/2$.  Note that $\alpha(1)=1/2$ so we are achieving exactly the stated approximation factor when $d=1$.  
\end{proof}

\section{Positive Hamiltonians}
\label{sec:positive-hamiltonians}

So far we have demonstrated an approximation algorithm which gets the optimal approximation factor for Quantum Max Cut Hamiltonians.  The next step is generalizing this to the case of $2$-local terms corresponding to rank-$1$ maximally entangled pure states and then to totally generic positive (strictly quadratic) $2$-local terms.  To precisely describe the distinctions between these cases we need to discuss ``cost'' and ``moment'' matrices.  For a particular $2$-local term $H_{ij}$ along some edge $ij$, the cost matrix $C_{ij}$ is defined as the matrix with coefficients corresponding to the Pauli basis elements.  Similarly, the $2$-local moment matrix $M_{ij}$ is set so that it's elements capture the relevant values from $\las_2$:

\begin{mdframed}
\begin{definition}\label{def:cost_matrix_generic}
Given $H=\sum_{ij\in E} w_{ij} H_{ij}$, let $c_{ij} \in \R^{3 \times 3}$ be defined as in \Cref{def:relaxed-objective} for each $ij\in E$.  Define $C_{ij}\in \mathbb{R}^{3 \times 3}$ with rows indexed by $\{X_i, Y_i, Z_i\}$ and columns indexed by $\{X_j, Y_j, Z_j\}$ such that:
\begin{equation}
    C_{ij}(\sigma_i, \tau_j)=c_{ij}(\sigma_i\tau_j) \text{  for $\sigma, \tau \in \{X, Y, Z\}$.}
\end{equation}

Let $\{\ket{\tau}\}_{\tau \in \mathcal{P}_n(2)}$ be some feasible solution to \Cref{def:lasserre}.  For each pair of vertices $i$, $j$ $\in [n]$ with $i\neq j$ define the moment matrix $M_{ij}\in  \mathbb{R}^{3 \times 3}$ with rows indexed by $\{X_i, Y_i, Z_i\}$ and columns indexed by $\{X_j, Y_j, Z_j\}$ such that:
\begin{equation}
    M_{ij}(\sigma_i, \tau_j)= \braket{\sigma_i | \tau_j} \text{  for $\sigma, \tau \in \{X, Y, Z\}$.}
\end{equation}

Let the value for each edge $ij\in E$ be defined as 
\begin{equation}
    v_{ij}=\frac{4 \tr[C_{ij}^T M_{ij}]}{3}
\end{equation}

\end{definition}
\end{mdframed}

In the current context, increasing the generality of the $2$-local Hamiltonian considered corresponds to increasing the generality of the {\it cost matrix} considered.  
\begin{mdframed}
\begin{definition}
Define
\begin{equation}\label{eq:def_S}
    \Sh:=conv\,\,\{(-1, -1, -1), (-1, 1, 1), (1, -1, 1), (1, 1, -1)\}.
\end{equation}
Let $H=\sum_{ij\in E} w_{ij} H_{ij}$ be some $2$-local Hamiltonian with $w_{ij}\geq 0$ for all $ij\in E$.  Let $C_{ij}$ be defined as in \Cref{def:cost_matrix_generic}.  
\begin{enumerate}
    \item  We say that $H$ belongs to the ``rank-$1$'' case if $C_{ij}=-U_{ij}/4$ for $U_{ij}\in SO(3)$ for all $ij$.
    \item  We say that $H$ belongs to the ``positive'' case if $C_{ij}=U_{ij} \Sigma_{ij} V_{ij}/4$ for $U, V\in SO(3)$, and for $\Sigma_{ij}$ diagonal with $diag(\Sigma_{ij})\in \Sh$
\end{enumerate}
\end{definition}
\end{mdframed}

The rank-$1$ case corresponds to $H_{ij}$ proportional to a rank $1$ projector onto a maximally entangled state (see \cite{G16} for some intuition).  The cost matrix can be seen to correspond to a locally rotated Bell state in this case.  The positive case captures all $H_{ij}\semigeq 0$ with only strictly quadratic terms.  Such a local term is the sum of rank-$1$ projectors from the former case, and each $C_{ij}$ in the positive case can be written as a convex combination of cost matrices from the rank-$1$ case.  The decomposition $C_{ij}=U \Sigma V$ can be thought of as a ``signed'' singular value decomposition (SVD) where we first calculate the rotation matrices on either side, then we absorb signs into the diagonal matrix to get $U, V\in SO(3)$ rather than $O(3)$.  

The algorithm is similar to \Cref{alg:main}, with the two main differences being how the large edges are treated and the rounding algorithm used on qubits not in $L$.  Since the local terms $H_{ij}$ may different for each edge $ij$, we must tailor the product state on edge $ij$ to the particular 2-Local term $H_{ij}$.  We can make use of a technique from~\cite{G12}.  For each edge $ij$ we will calculate the largest eigenvector for the local term $H_{ij}$ and use the Schmidt decomposition to construct a good product state.  %For qubits not in $L$, we cannot reduce the vectors $\{\ket{X_i}, \ket{Y_i}, \ket{Z_i}\}$ to a single vector $\ket{W_i}$ so a more general rounding algorithm is required~\cite{PT20}.  
\begin{mdframed}
\begin{algorithm}[Product-state approximation algorithm for strictly quadratic Hamiltonians]
\label{alg:generic}
\mbox{}\\[1ex]
\emph{Input}: A strictly quadratic $2$-local Hamiltonian $H=\sum_{ij\in E}w_{ij} H_{ij}$ (\Cref{def:strictly-quadratic-2-LH}), and a threshold $\gamma \in (-1/3,1]$.
\begin{enumerate}
    \item  Solve the instance of level 2 of the quantum Lasserre hierarchy, $\las_2$ (\Cref{def:vector-solution}) for $H$.
    \item  For each edge $ij$ calculate its value $v_{ij}$ from the $\las_2$ solution, according to \Cref{def:cost_matrix_generic} (or \Cref{def:relaxed-objective}).
    \item  Let $L:=\{ij\in E: v_{ij} > \gamma\}$ be the set of ``large'' edges, which are guaranteed to be a matching (\Cref{rem:matching}).  Let $B$ be the set of vertices not adjacent to a large edge: $B=\{i\in [n]: ij \notin L \text{ for any }j\}$.
    \item \label{item:execute-product-state-algorithm} 
    Execute the rounding algorithm of \cite{PT20} for qubits in $B$.  Denote the output Bloch vectors as: $\{(\theta_i^X, \theta_i^Y, \theta_i^Z)\}_{i \in B}$.
    \item  For each $ij\in L$: let $\ket{\psi_{ij}}=\text{argmax} \,\,\{\braket{\phi|H_{ij}|\phi}:\ket{\phi}\in \mathbb{C}^4,\,\braket{\phi|\phi}=1\}$.  Let $\ket{\psi_{ij}}=\sum_k \alpha_k \ket{w_k}\ket{v_k}$ be a Schmidt decomposition.  Define $\eta_{ij}=\sum_k \alpha_k^2\ket{w_k} \bra{w_k}\otimes \ket{v_k} \bra{v_k}$.
    \item  Output the state:
    \begin{equation}
        \rho = \prod_{ij \in L}\eta_{ij} \prod_{k \in B} \left( \frac{\mathbb{I}+\theta_k^X X_k + \theta_k^Y Y_k + \theta_k^Z Z_k}{2}\right)
    \end{equation}
    \end{enumerate}
\end{algorithm}
\end{mdframed}

The required analysis is semantically similar to the QMC case, we simply have to prove the generalizations of \Cref{fact:1} and \Cref{fact:2}.  For \Cref{fact:2} we are able to find rigorous generalizations for both the rank-$1$ and positive cases.  For \Cref{fact:1} we can only give computational evidence.  The issue is that the performance of the rounding algorithm on a single edge is no longer a (relatively) simple function of a single variable.  Instead, the analysis of \cite{PT20} presents it as an infinite series in which each term is a function of several parameters (\Cref{lem:hermite}).  Since we cannot exactly evaluate the series we cannot exactly express the approximation factor.  One option is to take the approach of \cite{PT20} and try to truncate the series to a sufficiently high order and get an approximation.  Unfortunately, since here we are analyzing a tight algorithm, the analysis has very little slack and a truncation will not serve our purposes unless it is of extremely high order (necessitating an intractable analysis).  As such we numerically optimize a very high order expansion to give evidence of the generalization of \Cref{fact:1} in \Cref{sec:numerics}.

\setcounter{factc}{0}
\begin{fact}
(generalized)  If $H=\sum_{ij\in E} w_{ij} H_{ij}$ for $w_{ij}\geq 0$ and $H_{ij}\geq 0$, then $v_{ij}\leq \gamma$ implies product state rounding has approximation factor $\geq 1/2$ on edge $ij$.
\end{fact}

\begin{fact}
(generalized)  If $H=\sum_{ij\in E} w_{ij} H_{ij}$ for $w_{ij}\geq 0$ and $H_{ij}\geq 0$, then $v_{ij} > \gamma$ implies $v_{ik}< 1/3$.
\end{fact}

\subsection{Generalization of \Cref{fact:2}}
The main purpose of this section is to prove and motivate the generalization of \Cref{fact:2}.  We will show that \Cref{fact:2} implies the generalization of \Cref{fact:2} to the rank-$1$ case which then implies the generalization to positive cost matrices.  To this end let us state and prove the first generalization.
\begin{theorem}\label{thm:star_bound_rank_1}
Let $C_{ij}=-U/4$ and $C_{ik}=-V/4$ be cost matrices with $U, V \in SO(3)$.  Let $M_{ij}$ and $M_{ik}$ be moment matrices from $\las_2$ with values $v_{ij}$ and $v_{ik}$ defined according to \Cref{def:cost_matrix_generic}.  If $v_{ij} > \gamma$ then $v_{ik} < 1/3$.
\end{theorem}
\begin{proof}
We can write
\begin{align}
    3 v_{ij}=4\tr[C_{ij} M_{ij}]=\tr[-U M_{ij}]=\tr[-\mathbb{I} \cdot{}  M_{ij} U]=\tr [-\mathbb{I} \cdot M_{ij}'] \text{  and }\\
    3 v_{ik}=4\tr[C_{ik} M_{ik}]=\tr[-V M_{ik}]=\tr[-\mathbb{I} \cdot{}  M_{ik} V]=\tr [-\mathbb{I} \cdot M_{ik}'].
\end{align}
Note that $M_{ij}'$ and $M_{ik}'$ are moment matrices with $SO(3)$ rotations on $j$ and $k$ respectively.  These correspond to local unitary rotations on $j$ and $k$ by \cite{H96}.  Hence, by \Cref{lem:rotation_L2}, $M_{ij}'$ and $M_{ik}'$ still correspond to moment matrices for a feasible solution to $\las_2$.  Since the cost matrices are the same as QMC, we may directly invoke \Cref{fact:2} (from \Cref{sec:tech_overview}).
\end{proof}
The next step is the generalization to the positive case:

\begin{theorem}
Let $C_{ij}=U_1 \Sigma_1 V_1/4$ and $C_{ik}=U_2 \Sigma_2 V_2/4$ be cost matrices with $U_\ell, V_\ell \in SO(3)$ and with $diag(\Sigma_\ell) \in \Sh$ for $\ell=1,2$.  Let $M_{ij}$ and $M_{ik}$ be moment matrices from $\las_2$.  Then, if $v_{ij} > \gamma$ then $v_{ik} < 1/3$.
\end{theorem}
\begin{proof}
Since the diagonal elements of $\Sigma_1$ are in $\Sh$, we can write $\Sigma_1=\sum_m \alpha_m \Delta_m$ where each $\Delta_m$ has on it's diagonal an extreme point of $\Sh$, $\sum_m \alpha_m=1$ and $\alpha_m\geq 0$.  

By inspection, see \Cref{eq:def_S}, it is clear that each $U_1 \Delta_m V_1$ can be written as $-W$ for $W_m$ $\in SO(3)$.  Hence, 
\begin{equation}
    3 v_{ij}=4\tr[C_{ij} M_{ij}]=\sum_m \alpha _m \tr[-W_m M_{ij}]
\end{equation}
Then $v_{ij} > \gamma$ implies that there exists some $m$ such that $\tr[-W_m M_{ij}] > 3 \gamma$.  If we apply the same kind of decomposition to $C_{ik}$ we obtain:
\begin{equation}
    3v_{ik}=4\tr[C_{ik} M_{ik}]=\sum_q \beta_q \tr[-R_q M_{ik}].
\end{equation}
Since there exists $m$ with $\tr[-W_m M_{ij}] > 3 \gamma$, \Cref{thm:star_bound_rank_1} implies that $ \tr[-R_q M_{ik}] < 1$ for all $q$, hence:
\begin{equation}
    v_{ik}=\frac{1}{3}\sum_q \beta_q \tr[-R_q M_{ik}] < 1/3.
    \end{equation}

\end{proof}

\subsection{Numerics}\label{sec:numerics}

Now we turn to the task of demonstrating generalizations of \Cref{fact:1} for the positive case.  The rank-$1$ case is a specialization (and would require the same techniques) so we will focus on the positive case.  Just as Parekh and Thompson \cite{PT20}, we will write down the approximation factor as a random variable then express it using a convergent series.  Since no analytic expression for the series is available to us, we can approximate the expectation to arbitrary precision by truncating the series and bounding the ``error term'' obtained from neglecting the remainder.  Contrary to \cite{PT20} we will use a very high order Hermite expansion and optimize over it computationally.  Their goal was to demonstrate a good approximation factor analytically, so a relatively low order expansion was used to keep the analysis tractable.  Demonstrating the generalization of \Cref{fact:1} requires a very precise estimate of the expectation necessitating a very high order expansion. We can reduce proving the generalization of \Cref{fact:1} in this case to an optimization over three parameters.  Hence we can find numerics which are fairly rigorous, although they fall short of a formal computer assisted proof. 

As previously discussed, $M_{ij}$ can be expressed as a Gram matrix of vectors from $\las_2$:
\begin{equation}
    M_{ij}=\begin{bmatrix}
    \, & \bra{X_i} & \, \\
     \, & \bra{Y_i}& \, \\
     \, & \bra{Z_i}& \, 
    \end{bmatrix}
    \begin{bmatrix}
    \, & \, & \, \\
    \ket{X_j} & \ket{Y_j} & \ket{Z_j}\\
    \, & \, & \, 
    \end{bmatrix}=: V_i^T V_j.
\end{equation}
If $\ket{r}$ is a randomly sampled standard multivariate normal random variable, then the expected approximation factor from \cite{PT20} can be written as:
\begin{equation}
    \alpha_{ij}=\frac{1+\mathbb{E}\left[\frac{\bra{r} V_i (4 C_{ij}) V_j^T \ket{r}}{||V_j^T \ket{r}|| \cdot ||V_i^T \ket{r}||} \right]}{1+3 v_{ij}}.
\end{equation}
Using exactly the ``SVD trick'' from \cite{PT20} (restated as \Cref{lem:svd} in the Appendix) we may assume that $M_{ij}$ is diagonal with $diag(M_{ij})=(a, b, c)$ and that $C_{ij}$ has diagonal elements $diag(C_{ij})=(p, q, r)$.  We may assume further that $(a, b, c)$ and $(p, q, r)\in \Sh$.  If we define the multivariate normal random vectors $z=(\braket{r|X_i}, \braket{r|Y_i}, \braket{r|Z_i})^T$ and $z'=(\braket{r| X_j}, \braket{r|Y_j}, \braket{r|Z_j})^T$, then we can express the approximation factor as:
\begin{equation}\label{eq:reduced_approx_factor}
\alpha_{ij}=\frac{1+\mathbb{E} \left[\frac{z^T\Sigma z'}{||z|| \cdot ||z'||} \right]}{1+a p + b q + c r}
\end{equation}
where $\Sigma$ is a diagonal matrix with $diag(\Sigma)=(p, q, r)$.  The generalization of \Cref{fact:1} to the case of interest then becomes:
\begin{mdframed}
\begin{lemma}[Computational]\label{lem:generic_fact2}
For $(a, b, c)$ and $(p, q, r)$ in $\Sh$,
\begin{equation}
    a p + b q + r c \leq 3 \gamma \Rightarrow \alpha_{ij} \geq 1/2.
\end{equation}
\end{lemma}
\end{mdframed}

Before delving into the details of our optimization we need to note some important simplifications.  Essentially the following shows that for fixed $(a, b, c)$ in the above, we need only check a few points to find a minimum among $(p, q, r)$.  While the optimization naively would have $6$ independent parameters, we reduce the problem to $3$.  
\begin{lemma}\label{lem:extreme_points}
Fix $(a, b, c)\in \Sh$.  For each $(p, q, r)$ consider $\alpha_{ij}$ as a function of $(p, q, r)$: $\alpha_{ij}^{(a, b, c)}(p, q, r)$.  Then, 
\begin{equation*}
    \min_{\substack{(p, q, r) \in \Sh\\ ap+bq+cr \leq 3 \gamma}} \alpha_{ij}^{(a, b, c)}(p, q, r) \geq \min_{(p, q, r) \in \mathcal{B}} \alpha_{ij}^{(a, b, c)}(p, q, r)
\end{equation*}
where $\mathcal{B}$ is the set of extreme points of the polytope formed by taking the intersection of $\Sh$ with the half-plane $\{(p, q, r): ap+bq+cr \leq 3 \gamma\}$.
\end{lemma}
\begin{proof}
Using the notation of random variables from \Cref{eq:reduced_approx_factor}, let us define:
\begin{equation*}
    A=\mathbb{E}\left[ \frac{z_1 z_1'}{||z|| \cdot{}||z'||}\right], \,\,\,\,\,\,B=\mathbb{E}\left[ \frac{z_2 z_2'}{||z|| \cdot{}||z'||}\right], \,\,\,\text{and }\,\,\, C=\mathbb{E}\left[ \frac{z_3 z_3'}{||z|| \cdot{}||z'||}\right].
\end{equation*}
Let $\mathcal{P}$ be the polytope formed by taking the intersection of $\Sh$ with the half-place $\{(p, q, r): ap+bq+cr \leq 3 \gamma\}$.  We are optimizing $(p, q, r)$ over this polytope.  Write $(p, q, r)$ as a convex combination of the extreme points of $\mathcal{P}$: $(p, q, r)=\sum_i \lambda_i (p_i, q_i, r_i)$.  Then,
\begin{align*}
     \alpha_{ij}(p, q, r)= \frac{1+p A+q B+r C}{1+p a + q b + r c}=\frac{1+\sum_i \lambda_i (A, B, C) \cdot (p_i, q_i, r_i)}{1+\sum_i \lambda_i (a, b, c)\cdot (p_i, q_i, r_i)}\\
     =\frac{\sum_i \lambda_i(1+  (A, B, C) \cdot (p_i, q_i, r_i))}{\sum_i \lambda_i(1+ (a, b, c)\cdot (p_i, q_i, r_i))}.
\end{align*}
Now apply the elementary fact that for non-negative constants $\{s_i, t_i\}$, $\frac{\sum_i s_i}{\sum_i t_i}\geq \min_{i:t_i \neq 0} \frac{s_i }{t_i}$.  It is easy to verify that $(1+  (A, B, C) \cdot (p_i, q_i, r_i))$ and $(1+ (a, b, c)\cdot (p_i, q_i, r_i))$ always correspond to $\tr[T \rho]$ for $\rho$ a physical density matrix and $T\semigeq 0$ (see \cite{H96, G16, PT20}).  Hence $(1+ (a, b, c)\cdot (p_i, q_i, r_i))$ and $(1+  (A, B, C) \cdot (p_i, q_i, r_i))$ are non-negative and we can apply the elementary fact:
\begin{equation*}
    \alpha_{ij}(p, q, r) \geq \frac{1+  (A, B, C) \cdot (p_i, q_i, r_i)}{1+ (a, b, c)\cdot (p_i, q_i, r_i)} \text{  for some $i$.}
\end{equation*}
\end{proof}
\noindent In light of the above lemma we will switch the notation $\alpha_{ij}(a, b, c)$ for the minimum value of $\alpha_{ij}^{(a, b, c)}(p, q, r)$ over the implicitly defined region $\mathcal{B}$:
\begin{equation}
    \alpha_{ij}(a, b, c):=\min_{(p, q, r) \in \mathcal{B}} \alpha_{ij}^{(a, b, c)}(p, q, r)
\end{equation}

In addition to \Cref{lem:extreme_points}, by exploiting symmetries in our problem we may assume the following:
\begin{lemma}
If $(a^*, b^*, c^*)={\rm argmin}_{(a, b, c) \in \Sh} \alpha_{ij}(a, b, c)$ then we may assume $a^*\leq b^*$, and $b^* \leq c^*$.  Further if $a^*+b^*+c^* > -1.5$ then $\alpha_{ij}(a, b, c) \geq 1/2$ hence me may restrict our search to $a^*+b^*+c^* \leq -1.5$. 
\end{lemma}
\begin{proof}
The first piece of the lemma is clear.  If the optimal assignment $(a^*, b^*, c^*)$ does not satisfy $a^* \leq b^* \leq c^*$ then we may shuffle them as well as the corresponding $(p^*, q^*, r^*)$ to obtain $a^* \leq b^* \leq c^*$ without effecting the approximation factor.  

The second piece follows from \cite{PT20}.  Essentially in the region corresponding to $a^*+b^*+c^* > -1.5$ we may use a very low order expansion to show the expectation is above $1/2$. We may lower bound $\alpha_{ij}$ as
\begin{equation}
    \min_{(a, b, c) \in \Sh} \alpha_{ij} (a, b, c) \geq \min_{\substack{(a, b, c) \in \Sh\\(p, q, r)\in \Sh}}\frac{1+p A + q B + r C}{1 + p a +q b + r c},
\end{equation}
using the notation of \Cref{lem:extreme_points}.  Using an argument very similar to \Cref{lem:extreme_points} (see \Cref{lem:pqr_sym} for a statement) we may lower bound this by setting $(p, q, r)=(-1, -1, -1)$.  
\begin{equation}
     \min_{(a, b, c) \in \Sh} \alpha_{ij} (a, b, c) \geq \min_{(a, b, c) \in \Sh}\frac{1- A -  B - C}{1 - a - b - c}.
\end{equation}
$A$, $B$ and $C$ are expectations which may be approximated via Hermite polynomial expansion using techniques from \cite{PT20}.  By \Cref{lem:hermite} and \Cref{lem:remainder} we can lower bound:
\begin{equation}
    -A-B-C \geq -\frac{8}{9 \pi} (a+b+c)-3(1/3-8/(9\pi)).
\end{equation}
Setting $x=(a+b+c)$ then we find
\begin{equation}
     \min_{(a, b, c) \in \Sh} \alpha_{ij} (a, b, c) \geq \min_{ x\in[-3, 1]}\frac{1-8/(9\pi) x-3(1/3-8/(9\pi))}{1-x}.
\end{equation}
The R.H.S.is monotonically increasing as a function of $x$ and the value of $x$ for which it exactly equals $1/2$ is $x\approx -1.61$.  Hence, (in particular) if $x>-1.5$ the R.H.S. must be larger than $1/2$.

\end{proof}

\paragraph{Computational Tests used to verify \Cref{lem:generic_fact2}.}  We make use of the previously described simplifications plus brute-force search to verify \Cref{lem:generic_fact2}.  We first create a mesh over the $(a, b, c)$ parameter space restricted to $a\leq b \leq c$ and $a+b+c \leq -1.5$.  The mesh itself will be a list of points each one of which is at most side length $\ell$ away from it's neighbor.  The mesh can be thought of as points on a square lattice with side length $\ell$.  For each point in the mesh $(a, b, c)$ we will find the extreme points guaranteed by \Cref{lem:extreme_points}: $(p, q, r) \in \mathcal{B}$.  

Then for all parameters $(a, b, c, p, q, r)$ found previously, we will evaluate a bound for:
\begin{equation}
    \frac{1+p A +q B+r C}{1+ p a +q b + r c }
\end{equation}
to very high precision using a very high order Hermite expansion (order $70$) from \cite{PT20}.  Recall $A$, $B$ and $C$ are defined in \Cref{lem:extreme_points}.  By including a remainder term, $rem$, essentially we will calculate a very large polynomial $t:\mathbb{R}^3 \rightarrow \mathbb{R}$ such that
\begin{equation}
     \frac{1+p A +q B+r C}{1+ p a +q b + r c } \geq \frac{1+p t(a, b, c) +q t(b, a, c)+r t(c, a, b)-rem }{1+ p a +q b + r c }
\end{equation}
Using the notation of \cite{PT20} and of \Cref{lem:hermite}, 
\begin{equation*}t(a, b, c)=\sum_{\substack{i, j \leq k\\i, j, k \leq 70}} \hat{f}_{i, jk}^2 (u_{i, jk}(a, b, c)+u_{i, jk}(b, a, c)+u_{i, jk}(c, a, b)).
\end{equation*}

By checking every point in the mesh we arrive at:
\begin{equation}
\alpha_{ij}^*=\min_{(a, b, c, p, q, r)\in \mathcal{A}}  \frac{1+p t(a, b, c) +q t(b, a, c)+r t(c, a, b)-rem }{1+ p a +q b + r c },
\end{equation}
where $\mathcal{A}$ is the set of parameters $(a, b, c, p, q, r)$ arrived at with the previous procedure.  

Hence by checking every point in the mesh we can be certain of the lower bound up to {\it spatial} accuracy determined by $\ell$.  As long as $\alpha_{ij}^*$ is greater than $1/2$ and the function does not change too rapidly on the length scale $\ell$, we can find a good estimate of the lower bound.  In the mesh, we use a side length $\ell=5 \cdot 10^{-3}$ when $-1 \leq (a+b+c)/3 \leq -0.85$ and, $\ell=5\cdot 10^{-2}$ when $-0.85 \leq (a+b+c)/3 \leq -0.5$.  The reason for this is that we anticipate the minimum value to occur in the region corresponding to $-1 \leq (a+b+c)/3 \leq -0.85$, so we use a more accurate mesh to better capture the function there.  The observed minimum was at $(a, b, c)=(-\gamma, -\gamma, -\gamma)$ and we obtained $\alpha_{ij}^*=0.500296$.  While the length scale is larger than the function difference $\alpha_{ij}^*-1/2\approx 10^{-4}$, the observed maximum magnitude of the gradient near the minimum is $\approx 10^{-3}$, hence we believe we have accurately captured the lower bound $\alpha_{ij}^*\geq  0.500296- 5\cdot 10^{-3} \cdot 10^{-3} >1/2$.  By magnitude of gradient here we mean $|\nabla_{a, b, c} \alpha_{ij}(a, b, c)|$, calculated numerically.

\section{Relevant Results from Previous Work} \label{sec:lems_from_beating_random}

Some of the work we present here relies on work by Parekh and Thompson \cite{PT20}.  Hence we provide ``clean'' descriptions here of the facts we need which follow from that work.  First define the following convex hull:
Let us also define the discrete delta function:
\begin{equation}
    \delta_{ij}=\begin{cases} 1 \text{ if $i=j$ }\\ 0 \text{  otherwise} \end{cases}.
\end{equation}

%With these definitions in hand we may define the important lemmas concerning the structure of moment matrices and cost matrices.  These facts were first proven in \cite{H96}:
%\begin{lemma}\label{lem:moment_props}
%Let $G=1/4 \sum_{l, m=0}^3 \Gamma_{lm} \sigma^l \otimes \sigma^m$ be some $2$-qubit projector of rank $1$, and let $\rho=1/4 \sum_{l, m=0}^3 \Delta_{lm} \sigma^l \otimes \sigma^m$ be a density matrix on $2$ qubits.  Let $M_G$ and $M_\rho$ be the corresponding moment matrices.  
%\begin{enumerate}
%\item  Let $A_1$, $A_2$ $\in SO(3)$.  Then, $M_G'=A_1 M_G A_2^T$ (resp. $M_\rho'=A_1 M_\rho A_2^T$) is a valid moment matrix for $U_1 \otimes U_2 G U_1^\dagger \otimes U_2^\dagger$ (resp. $U_1 \otimes U_2 \rho U_1^\dagger \otimes U_2^\dagger$)  for some local unitaries $U_1$, $U_2$.  Conversely, if $U_1$, $U_2$ are arbitrary unitaries then $A_1$, $A_2$ may be found so that $M_G'=A_1 M_G A_2^T$ (resp. $M_\rho'=A_1 M_\rho A_2^T$) are the moment matrices corresponding to $U_1 \otimes U_2 G U_1^\dagger \otimes U_2^\dagger$ (resp. $U_1 \otimes U_2 \rho U_1^\dagger \otimes U_2^\dagger$).
%\item  $[\Gamma_{11}, \Gamma_{22}, \Gamma_{33}] \in \mathcal{S}$.
%\end{enumerate}
%\end{lemma}

The next important piece is the SVD technique used to reduce the expectation of the standard product state rounding algorithm to a canonical form:

\begin{lemma}[\cite{PT20}]\label{lem:svd}
Let $V_1=[\mathbf{v}_{11}, \mathbf{v}_{12}, \mathbf{v}_{13}]\in \mathbb{R}^{d\times 3}$ and $V_2=[\mathbf{v}_{21}, \mathbf{v}_{22}, \mathbf{v}_{23}]\in \mathbb{R}^{d\times 3}$ be real matrices with normalized columns ($|| \mathbf{v}_{ij}||^2=1$) such that $V_1^T V_2$ is a valid moment matrix for some density matrix.  Let $C$ be the cost matrix for a rank $1$ strictly quadratic projector,  and let $\mathbf{r}\sim \mathcal{N}(0, \mathbb{I_d})$ be the same size as the vectors $\mathbf{v}_{ij}$.  Then,
\begin{align}\label{eq:22}
\mathbb{E}_{\mathbf{r}}\left[\frac{ \mathbf{r}^T V_1 (4C) V_2^T \mathbf{r}}{||V_1^T \mathbf{r}|| \,\, ||V_2^T \mathbf{r}||}\right]=\mathbb{E} \left[ \frac{p z_1 z_1' +q z_2 z_2' +r z_3 z_3'}{\sqrt{(z_1^2+z_2^2+z_3^2)((z_1')^2+(z_2')^2+(z_3')^2)}}\right]\\
{\rm and }\,\, Tr\left( V_1 (4 C) V_2^T \right)=a p+b q + c r\label{eq:23}.
\end{align}
\noindent with $[z_1, z_2, z_3, z_1', z_2', z_3']$ multivariate normal random variables with mean zero and covariance matrix
\begin{equation*}
\Sigma=\begin{bmatrix}  
1 & 0 & 0 & a & 0 & 0\\
0 & 1 & 0 & 0 & b & 0\\
0 & 0 & 1 & 0 & 0 &c\\
a & 0 & 0 & 1 & 0 & 0\\
0 & b & 0 & 0 & 1 & 0\\
0 & 0 & c & 0 & 0 & 1
\end{bmatrix}
\end{equation*}

Further, we can assume without loss of generality that $(a, b, c)$, and $(p, q, r)\in \Sh$.
\end{lemma}

One of the important technical observations in the previous work was that this can be evaluated with a convergent series:
\begin{lemma}[\cite{PT20}]\label{lem:hermite}
Let $[z_1, z_2, z_3, z_1', z_2', z_3']$ be as described in \Cref{lem:svd}.  Define the constants:
\begin{equation}\label{eq:50}
    \hat{f}_{i, jk}=\begin{cases} 2\sqrt{\frac{2}{\pi}} \frac{(-1)^p\sqrt{i! j! k!}}{(i-1)!! j!! k!! (1+2p)(3+2p)} \text{  if  $i$ is odd, $j$ is even, $k$ is even, and $i$, $j$, $k\in \mathbb{Z}_{\geq 0}$}\\
    0 \text{  otherwise  }\end{cases}
\end{equation}
for $p=(i+j+k-1)/2$.  Further define the following family of polynomials:
\begin{equation}
    u_{i, jk}(a, b, c) =\begin{cases}   
a^i b^j c^k+a^i b^k c^j \,\,\,\, \text{  if $k\neq j$ }\\
a^i b^j c^j \,\,\,\, \text{  if $k=j$}\end{cases}.
\end{equation}
If $(a, b, c)\in \Sh$, \textbf{then}
\begin{equation}
    \mathbb{E} \left[ \frac{z_1 z_1'}{||(z_1, z_2, z_3)|| \,\, ||(z_1', z_2', z_3')||}\right]=\sum_{i, j\leq k} \hat{f}_{i, jk}^2 u_{i, jk}(a, b, c).
\end{equation}

\end{lemma}

Since the above expresses our desired expectation in terms of a convergent series at each point, we can truncate the series to approximate the expectation to any precision we want:
\begin{lemma}[\cite{PT20}]\label{lem:remainder}
Let $A$ be a finite subset of $\mathbb{Z}_{\geq 0} \times \mathbb{Z}_{\geq 0} \times \mathbb{Z}_{\geq 0}$ such that $k\leq l$ for all $(j, k, l)\in A$ and let $(a_1, a_2, a_3)\in \Sh$.  Then, 
\begin{align*}
\left|\sum_{\substack{i, j, k\\j\leq k}}\hat{f}_{i, jk}^2 u_{i, jk}(a, b, c)-\sum_{(i, j, k)\in A}\hat{f}_{i, jk}^2 u_{i, jk}(a, b, c)\right| \leq  1/3-\sum_{(i, j, k) \in A}2^{1-\delta_{j, k}}\hat{f}_{i, jk}^2.
\end{align*}
\end{lemma}

The final important fact concerns the case when $a=b=c$.  For this, we get an exact expression in terms of the Hypergeometric function $\,_2 F_1$:

\begin{lemma}[\cite{B10, G19}]\label{lem:gp_exp}
Let $[z_1, z_2, z_3, z_1', z_2', z_3']$ be as described in \Cref{lem:svd} with $a=b=c$.  Then we have:
\begin{equation}
    \mathbb{E} \left[ \frac{(z_1, z_2, z_3) \cdot (z_1', z_2', z_3')}{||(z_1, z_2, z_3)|| \,\, ||(z_1', z_2', z_3')||}\right]=\frac{8a}{3\pi}\,_2 F_1\left[1/2, 1/2;5/2;a^2 \right]=: F[a]
\end{equation}
\end{lemma}

\begin{lemma}[\cite{PT20}]\label{lem:pqr_sym}
Fix $(a, b, c)$ and let $z$ and $z'$ be defined according to \Cref{lem:svd}.  Define:
\begin{gather*}
A= \mathbb{E} \left[ \frac{z_1 z_1'}{\sqrt{(z_1^2+z_2^2+z_3^2)((z_1')^2+(z_2')^2+(z_3')^2)}}\right], \\
\nonumber B= \mathbb{E} \left[ \frac{z_2 z_2'}{\sqrt{(z_1^2+z_2^2+z_3^2)((z_1')^2+(z_2')^2+(z_3')^2)}}\right], \\
\nonumber {\rm and } \,\,C= \mathbb{E} \left[ \frac{z_3 z_3'}{\sqrt{(z_1^2+z_2^2+z_3^2)((z_1')^2+(z_2')^2+(z_3')^2)}}\right].
\end{gather*}
Then, 
\begin{equation*}
\min_{(p, q, r)\in \Sh} \frac{1+p A+ q B + r C}{1 + p a +q b + r c} =\frac{1-A-B-C}{1-a-b-c} 
\end{equation*}
\end{lemma}

\end{document}